\documentclass[10pt,journal,final,twocolumn]{IEEEtran}
\IEEEoverridecommandlockouts

\usepackage{amsmath,amssymb,amsfonts}
\usepackage{cite}
\usepackage{graphicx}
\usepackage{epstopdf}

\usepackage{booktabs}
\usepackage{graphicx}
\usepackage{textcomp}
\usepackage{xcolor}
\usepackage{times}
\usepackage{subfigure}
\usepackage{amssymb,amsmath}
\usepackage{acronym}  
\usepackage{balance}

\usepackage{bm}         
\usepackage{algorithm} 
\usepackage{algorithmic}

\usepackage{lipsum}
\usepackage{stfloats}

\usepackage{url}
\usepackage{mathrsfs}
\usepackage{bbm}

\usepackage{amsthm}
\usepackage{enumerate}

\usepackage{hyperref}
\hypersetup{hidelinks, 
colorlinks=true,
allcolors=black,
pdfstartview=Fit,
breaklinks=true}

\theoremstyle{definition}

\newtheorem{lemma}{\bf Lemma}
\newtheorem{corollary}{\bf Corollary}

\theoremstyle{remark}
\newtheorem{remark}{\bf Remark}

\acrodef{ofdm}[OFDM]{orthogonal frequency division multiplexing}%
\acrodef{miso-ofdm}[MISO-OFDM]{multi-input single-output orthogonal frequency division multiplexing}%

\acrodef{ris}[RIS]{reconfigurable intelligent surface}%
\acrodef{qos}[QoS]{quality of service}%
\acrodef{idft}[IDFT]{inverse discrete Fourier transform}%
\acrodef{dft}[DFT]{discrete Fourier transform}%
\acrodef{cp}[CP]{cyclic prefix}%
\acrodef{csi}[CSI]{channel state information}%
\acrodef{awgn}[AWGN]{additive white Gaussion noise}%

\acrodef{qcqp}[QCQP]{quadratically constrained quadratic program}%
\acrodef{qp}[QP]{quadratic program}%
\acrodef{bs}[BS]{base station}%
\acrodef{ap}[BS]{base station}%
\acrodef{aps}[APs]{access points}%
\acrodef{qos}[QoS]{quality of service}%
\acrodef{ue}[UE]{user equipment}%
\acrodef{snr}[SNR]{signal-to-noise ratio}%
\acrodef{mmwave}[mmWave]{millimeter-wave}%
\acrodef{snr}[SNR]{signal-to-noise ratio}%

\acrodef{sinr}[SINR]{signal-to-interference-plus-noise ratio}%

\acrodef{ser}[SER]{symbol error rate}%
\acrodef{rc}[RC]{reflection coefficient}%
\acrodef{uavs}[UAVs]{unmanned aerial vehicles}%
\acrodef{mimo}[MIMO]{multiple-input multiple-output}%
\acrodef{noma}[NOMA]{non-orthogonal multiple access}%

\acrodef{ace}[ACE]{adaptive cross-entropy}%
\acrodef{wsr}[WSR]{weighted sum-rate}%
\acrodef{udn}[UDN]{ultra-dense network}%
\acrodef{Udn}[UDN]{Ultra-dense network}%

\def\BibTeX{{\rm B\kern-.05em{\sc i\kern-.025em b}\kern-.08em
    T\kern-.1667em\lower.7ex\hbox{E}\kern-.125emX}}

\begin{document}
\title{Enabling More Users to Benefit from Near-Field Communications: From Linear to Circular Array}
\author{\IEEEauthorblockN{Zidong Wu, Mingyao Cui, and Linglong Dai, \emph{Fellow, IEEE}}

\thanks{
This work was supported in part by the National Key Research and Development Program of China under Grant 2020YFB1807201, and in part by the European Commission through the H2020-MSCA-ITN META WIRELESS Research Project under Grant 956256. \emph{(Corresponding author: Linglong Dai.)}

The authors are with the Department of Electronic Engineering, Tsinghua University, Beijing 100084, China, and also with Beijing National Research Center for Information Science and Technology (BNRist), Beijing 100084, China (e-mails: wuzd19@mails.tsinghua.edu.cn; cuimy16@tsinghua.org.cn; daill@tsinghua.edu.cn).}
}

\maketitle
\vspace{-4em}
\begin{abstract}
Massive multiple-input multiple-output (MIMO) for 5G is evolving into the extremely large-scale antenna array (ELAA) to increase the spectrum efficiency by orders of magnitude for 6G communications. ELAA introduces spherical-wave-based near-field communications, where channel capacity can be significantly improved for single-user and multi-user scenarios. Unfortunately, the near-field region at large incidence/emergence angles is greatly reduced with the widely studied uniform linear array (ULA). Thus, many randomly distributed users may fail to benefit from near-field communications. In this paper, we leverage the rotational symmetry of uniform circular array (UCA) to provide uniform and enlarged near-field regions at all angles, enabling more users to benefit from near-field communications. Specifically, by exploiting the geometrical relationship between UCA and users, the near-field beamforming technique for UCA is developed. Based on the analysis of near-field beamforming, we reveal that UCA is able to provide a larger near-field region than ULA in terms of the effective Rayleigh distance. Moreover, a concentric-ring codebook is designed to realize efficient codebook-based beamforming in the near-field region. In addition, we find out that UCA could generate orthogonal near-field beams along the same direction when the focal point of the near-field beam is exactly the zeros of other beams, which has the potential to further improve spectrum efficiency in multi-user communications compared with ULA. Simulation results are provided to verify the effectiveness of theoretical analysis and feasibility of UCA to enable more users to benefit from near-field communications by broadening the near-field region.
\end{abstract}

\section{Introduction}\label{sec: intro}
\par With the increasing demand for data transmission, spectrum efficiency has become one of the most important key performance indicators (KPIs) for wireless communications~\cite{Navrati'16'}. In current 5G communications, massive multiple-input multiple-output (MIMO) has played a central role in increasing spectrum efficiency by orders of magnitude~\cite{Marzetta'14'j}. To meet higher requirements of spectrum efficiency for 6G, extremely large-scale MIMO (XL-MIMO) or extremely large-scale antenna array (ELAA) has attracted tremendous attention, which is featured by ten times more antennas at the base station (BS) compared with 5G massive MIMO systems~\cite{Zengyong'22'twc}.

\par The dramatically increased array aperture introduces near-field communications in ELAA systems~\cite{Marzetta'22'twc}. Specifically, in classical massive MIMO systems for 5G, since the antenna array at the BS is not very large, the electromagnetic wave modeled by near-field spherical wave could be simplified into the one modeled by far-field planar wave, which contributes to the concise system analysis and simplified system design~\cite{Ayach'14'j}. Nevertheless, since the phase discrepancy between the planar-wave and the spherical-wave model could not be ignored anymore due to the significantly enlarged array aperture, the widely-adopted far-field planar-wave model has become inaccurate for ELAA systems~\cite{cui'22'm}. Thus, the accurate near-field spherical-wave model has to be adopted instead in ELAA systems. To sum up, the expansion of the array size inevitably leads to the fundamental change in the electromagnetic propagation environment in 6G ELAA systems.
\vspace{-1em}
\subsection{Prior Works}\label{sec:intro prior}
\par Existing works on near-field communications can be mainly classified into two categories: Alleviating the performance degradation caused by the near-field effect and exploiting the near-field effect to enhance the system performance~\cite{cui'22'm}.
\par For the first category, due to the mismatch of the near-field propagation model and far-field communication techniques, many works were investigated to overcome the severe performance degradation of existing far-field techniques in the near-field region, such as channel state information (CSI) acquisition~\cite{Cui'22'tcom,han'20'j,You'22'beam,Xiu'22'j,Wen'22'jstsp}. The traditional angular-domain representation for far-field channels experiences serious power leakage in the near-field region, greatly reducing the accuracy of channel estimations~\cite{Cui'22'tcom}. Specifically, due to the adopted planar-wave model, the angular-domain representation could be employed to capture the angular sparsity of far-field channels. Plenty of algorithms including compressive sensing and deep learning have been proposed to exploit the angular sparsity to realize low-overhead channel estimations~\cite{Alkhateeb'14'ce,Gaozhen'21'jsac}. Nevertheless, it has been demonstrated that the angular sparsity does not hold anymore with near-field spherical-wave models, making far-field channel estimation methods fail to recover the near-field channels accurately~\cite{han'20'j}. To address this problem, a novel polar-domain representation was proposed in~\cite{Cui'22'tcom}, which performed uniform angular sampling and non-uniform distance sampling to successfully retrieve the sparsity of near-field channels. This representation method was also employed to perform fast beam training in the near-field region~\cite{You'22'beam}. In addition, it has been revealed that the circular array could be leveraged to reduce the size of near-field codebook, which is promising to reduce beam training overhead in near-field communications~\cite{ningboyu'23'wcl}.
\par After obtaining accurate CSI, efficient beamforming tailored for near-field models should also be performed to replace classical far-field beamforming in the near-field region to avoid beamforming loss~\cite{emil'21,zhang'22'j,zhanghaiyang'22'm}. Precisely, beamforming is realized by compensating for the phase differences at different antennas to achieve the constructive superposition of electromagnetic waves. It was revealed in~\cite{cui'22'm} that due to the remarkable phase discrepancy between planar waves and spherical waves, the far-field beamforming fails to compensate for the signal phases in near-field regions, where the constructive superposition could not be formed to obtain ideal beamforming gain. To cope with this problem, several recent works have adopted the spherical-wave model, instead of planar-wave model, to perform near-field beamforming to avoid inaccurate beamforming in the study of antenna designs~\cite{Waldschmidt'21'tmtt,chang'20'tap,Mei'22'tap}. In~\cite{Ribeiro'21'conf}, low-complexity zero-forcing precoding algorithms were developed to cope with the overwhelming complexity problem in fully-digital near-field communication systems, where the algebraic property of near-field channels was exploited to reduce the dimension of precoding matrices. Moreover, the property of near-field beamforming from the perspective of communication was investigated in~\cite{zhang'22'j}. In addition to the focusing ability in the angular domain as far-field beamforming,~\cite{zhang'22'j} pointed out that near-field beamforming is capable of focusing the signal power on a specific distance, which was also termed \emph{near-field beamfocusing}. In~\cite{zhanghaiyang'22'm}, the near-field beamforming techniques are further generalized into wireless power transfer scenarios.

\par In the second category, which is also the main focus of this paper, many works focused on the new opportunities brought by near-field communications. One opportunity is the increased spatial degrees of freedom (DoFs) for single-user communications. For far-field planar-wave models, since the line-of-signal (LoS) path of single-user MIMO only covers one direction, it was illustrated that this LoS channel matrix is rank-one and could only support one data stream~\cite{Sayeed'2013,Zijian'22'icc}. Surprisingly, it was proved that owing to the spherical-wave model, the near-field LoS path could cover a large range of directions, which may contribute to $10 \times$ more DoFs and data streams in the near-field region~\cite{yan'22'j}. To efficiently exploit the increased DoFs to improve spectrum efficiency, a dynamic hybrid precoding architecture was developed in~\cite{Wu'22'dap}, where the number of radio frequency (RF) chains could be designed to adapt to the DoFs for capacity enhancement. The second opportunity of near-field communications is that distance information could serve as a new utilizable dimension for multi-user communication capacity enhancement. Specifically, the classical spatial division multiple access (SDMA) employs directional beams to serve users at different angles, which fails to simultaneously serve users in the same direction with the pure LoS channel assumption. In contrast, owing to the enhanced focusing property of near-field beams, users at different distances could also be served and the concept of location division multiple access (LDMA) was proposed to significantly increase the spectrum efficiency~\cite{Wu'22'jsac}.  

\par On account of the significant benefits of near-field communications, it is desired to create near-field propagation environments for more users to significantly improve the system performance. Unfortunately, it has been noticed that the near-field region dramatically shrinks at large incidence/emergence angles of the widely-adopted uniform linear array (ULA)~\cite{cui'21}. The reduced near-field region is caused by the decreased effective array aperture at large angles, which results in a smaller phase difference between planar-wave and spherical-wave models. Therefore, near-field communications become prominent only in the very near region at large angles, resulting in many users failing to take advantage of the benefits of near-field communications. This phenomenon severely limits the performance improvement of ELAA systems. Therefore, how to enable more users to benefit from near-field communications becomes a critical problem.
\vspace{-1em}
\subsection{Our Contributions}\label{sec:intro contr}
\par To tackle this problem, in this paper we consider a variant array geometry, uniform circular array (UCA), to enlarge the near-field region, making it possible for more users to benefit from near-field communications\footnote{Simulation codes are provided to reproduce the results presented in this article: http://oa.ee.tsinghua.edu.cn/dailinglong/publications/publications.html.}. The contributions of this paper can be summarized as follows.
\begin{itemize}
	\item We exploit the rotational symmetry of UCA to explore the feasibility of changing the array geometry to enlarge near-field regions for the first time. In fact, the near-field region is determined by the effective array aperture observed from users. To avoid the near-field shrinkage problem at large incidence/emergence angles of ULA, we employ the UCA to generate an angle-invariant near-field region by utilizing its feature that the effective array aperture of UCA remains uniform at any angle. Following the concept of effective Rayleigh distance (ERD) to define the near-field region~\cite{cui'21}, we verify that the near-field region of UCA significantly expands at large angles compared with ULA. According to the fact that users are randomly distributed in a cell, many users are likely to locate in large angles. Therefore, the enlarged near-field region could enable more users to benefit from near-field communications.
	\item By exploiting the geometrical relationship between UCA and users, the near-field beamforming mechanism of UCA is theoretically analyzed. Different from existing works on UCA which only consider planar-wave propagation models~\cite{Kallnichev'2001,Hussain'2005,Zhang'17'tawp}, the beamfocusing property of UCA based on spherical-wave model is investigated. We prove that the near-field beamforming gain of UCA is in the form of Bessel functions in both angle and distance domains. The derived closed-form beamforming gain reveals that the far-field beamforming techniques for UCA suffer from severe performance loss in the near-field region, or even has several {\emph{zero}} beamforming gain at some specific locations. These findings highlight the uniqueness of UCA near-field beamforming.
	\item Based on the UCA near-field beamforming mechanism, more near-field communication techniques are investigated. To design the near-field codebook for UCA to realize codebook-based beamforming, an efficient sampling method is performed in both angle and distance domains to form the concentric-ring codebook. Specifically, according to the rotational symmetry of UCA, we reveal that the distance sampling remains the same for any angle. Exploiting this property, the selected samples form the shape of multiple concentric rings, the radius of which is determined by the distance sampling rule. Therefore, the concentric-ring codebook for UCA can be constructed with near-field beamforming vectors focused energy on those sampling points. In addition, noticing the zero beamforming gain in the near-field region, we find out that UCA could generate orthogonal near-field beams along the same direction, which has the potential to mitigate interferences for spectrum efficiency enhancement compared with ULA.
	\item Moreover, beamforming for a more practical version of UCA, cylindrical array, is discussed in this paper. We derive a closed-form near-field beamforming gain for the cylindrical array when users are located in the central plane of the cylindrical array. Through simulations, we verify that UCA could provide an enlarged near-field region compared with ULA and the effectiveness of the near-field beamforming analysis is also validated.
\end{itemize}
\vspace{-1em}
\subsection{Organization and Notation}\label{sec:intro org}
\textit{Organization}: The remainder of the paper is organized as follows. Section~\ref{sec: sys} introduces the system model and near-field beam focusing vectors for UCA. Section~\ref{sec: beamform} investigates the near-field beamforming techniques for UCA systems in both angular and distance domains in comparison to ULA. Then, the near-field concentric-ring codebook for UCA is provided in~\ref{sec: codebook}. The generalization from UCA to cylindrical arrays is discussed in Section~\ref{sec: cyl}. Simulation results are provided in Section~\ref{sec: sim} and conclusions are drawn in Section~\ref{sec: con}.

\textit{Notations}: Lower-case and upper-case boldface letters represent vectors and matrices, respectively; $\mathbb{C}$ denotes the set of complex numbers; ${[\cdot]^{-1}}$, ${[\cdot]^{T}}$ and ${[\cdot]^{H}}$ denote the inverse, transpose and conjugate-transpose, respectively; $|\cdot|$ denotes the norm of its complex argument; $\mathcal{CN}(\mu,\sigma^2)$ denotes the complex Gaussian distribution with mean $\mu$ and variance $\sigma^2$.
\vspace{-1em}
\section{System Model}\label{sec: sys}
\par We consider a downlink transmission between the BS equipped with $N$-element UCA and single-antenna users. For illustration simplicity, we assume that the user is restricted in the 2D space of the same plane of UCA in this paper unless otherwise stated. Following exemplary UCA models in~\cite{Tewfik'90'c,Gentile'08'icc}, the antennas are uniformly distributed along the circle with a radius of $R$. For ease of expression, the polar coordinate is adopted to describe the geometrical relationships. The antennas can be represented with the location indexed by $(R, \psi_n)$, where $\psi_n$ is defined as $\psi_n = \frac{2\pi n}{N}$ for $n = 1,\cdots,N$. Then, the received signal $y$ at the user is expressed as
\begin{equation}
\label{eq: received signal}
\begin{aligned}
y = {\bf{h}}^H {\bf{f}} s + n,
\end{aligned}
\end{equation}
where ${\bf{h}} \in \mathbb{C}^{N \times 1}$ , ${\bf{f}} \in \mathbb{C}^{N \times 1}$, $s \in \mathbb{C}$, and $n \sim \mathcal{CN}(0,\sigma_n^2)$ denote the wireless channel, beamforming vector, transmitted signal, and additive white Gaussian noise (AWGN), respectively. In the classical massive MIMO systems where the far-field planar-wave model can be applied, the channel could be modeled by~\cite{Zhang'17'tawp}
\begin{equation}
\label{eq: far-field channel}
\begin{aligned}
{\bf{h}}_{\rm{far}} = \sqrt{\frac{N}{L}} \sum_{l=1}^{L} \alpha_l {\bf{a}}(\phi_l),
\end{aligned}
\end{equation}
where $L$ denotes the number of paths. Moreover, $\alpha_l$ and $\phi_l$ denote the complex path gain and the angle of the $l^{\rm{th}}$ path, respectively. In this paper, we mainly focus on the line-of-sight (LoS) component of channels, which is dominant and contributes most to communications, especially in high-frequency bands, such as millimeter-wave (mmWave) and terahertz (THz)~\cite{Ayach'14'j}. Therefore, the number of paths $L$ is set to $1$ throughout this paper. Then, the channel could be directly built with the beam steering vector ${\bf{a}}(\phi_l)$ under the far-field planar-wave propagation model as~\cite{Kallnichev'01}
\begin{equation}
\label{eq: far beam vector}
\begin{aligned}
{\bf{a}}(\phi) = \frac{1}{\sqrt{N}} \left[e^{j\frac{2\pi}{\lambda}R\cos(\phi-\psi_1)},\cdots, e^{j\frac{2\pi}{\lambda}R\cos(\phi-\psi_N)} \right]^T.
\end{aligned}
\end{equation}
\par Nevertheless, with the increasing number of antennas in ELAA, the far-field planar-wave propagation model is no longer valid in the near-field region, where the spherical-wave propagation model has to be adopted to characterize the propagation of electromagnetic waves~\cite{cui'22'm}. Traditionally, the concept of Rayleigh distance is defined as the boundary of near- and far-field regions, which is expressed as~\cite{Sherman'62'j}
\begin{equation}
\label{eq: rayleigh}
\begin{aligned}
r_{\rm{RD}} = \frac{2D^2}{\lambda},
\end{aligned}
\end{equation}
where $D$ denotes the array aperture and $\lambda$ denotes the wavelength. Note that the increased number of antennas will yield a larger Rayleigh distance. For instance, the $256$-element ULA operating at 30 GHz leads to the Rayleigh distance of about $300\,{\rm{m}}$, indicating that most users in a cell are located in the near-field region. This example illustrates that near-field communications happen frequently in ELAA systems, and thus the near-field channel model has to be adopted.
\begin{figure}[!t]
    \centering
    \setlength{\abovecaptionskip}{0.cm}
    \includegraphics[width=3in]{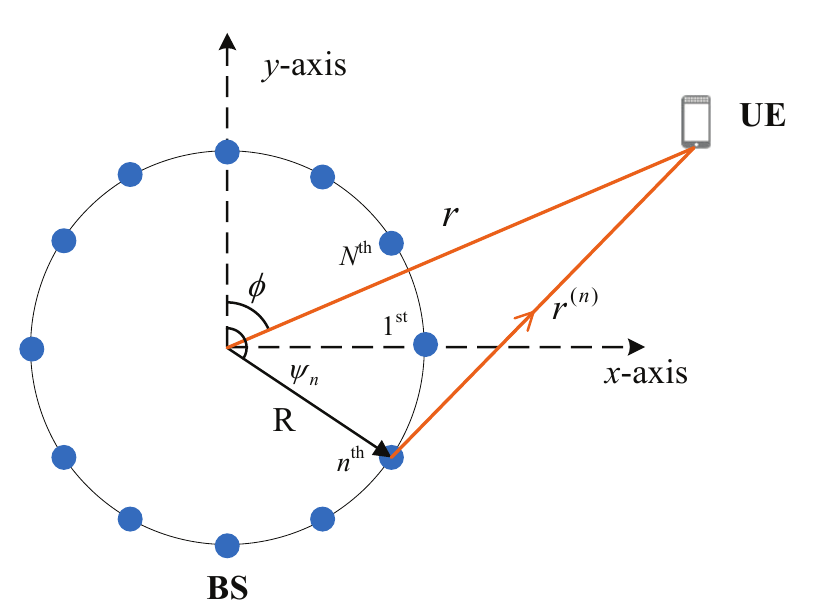}
    \caption{The geometrical relationship between UCA and user in the near-field region.}
    \label{img: UCA strcture}
\end{figure}
\par Based on the spherical-wave propagation model, the near-field channel could be expressed similarly as
\begin{equation}
\label{eq: near-field channel}
\begin{aligned}
{\bf{h}}_{\rm{near}} = \sqrt{\frac{N}{L}} \sum_{l=1}^{L} \alpha_l {\bf{b}}(r_l, \phi_l),
\end{aligned}
\end{equation}
where the classical {\emph{far-field beam steering vector}} ${\bf{a}}(\phi)$ is replaced by {\emph{near-field beam focusing vector}} ${\bf{b}}(r, \phi)$, which is capable of focusing the signal power on the position $(r, \phi)$~\cite{Cui'22'tcom}. Again, the number of paths $L=1$ is set. It is worth noting that, although the LoS channel is the main focus of this paper, the near-field determination method can also be applied to scatterers in non-line-of-sight (NLoS) channels. A scatterer is distinguished in the near-field region if the BS-scatterer distance does not exceed the Rayleigh distance. Correspondingly, similar beamfocusing vectors could be employed to model NLoS channels. As shown in Fig.~\ref{img: UCA strcture}, the near-field beam focusing vector can be described as
\begin{equation}
\label{eq: near beam vector}
\begin{aligned}
{\bf{b}}(r, \phi) = \frac{1}{\sqrt{N}} \left[e^{-j\frac{2\pi}{\lambda}(r^{(1)}-r)},\cdots, e^{-j\frac{2\pi}{\lambda}(r^{(N)}-r)} \right]^T,
\end{aligned}
\end{equation}
where $r^{(n)}$ denotes the propagation distance between the user and the $n^{th}$ antenna of UCA. The propagation distance $r^{(n)}$ can be formulated as
\begin{equation}
\label{eq: distance near}
\begin{aligned}
r^{(n)} & = \sqrt{r^2+R^2-2rR\cos(\phi-\psi_n)}\\
& \mathop {\approx}\limits^{(a)} r - R\cos(\phi-\psi_n) + \frac{R^2}{2r}\left(1-\cos^2(\phi-\psi_n)\right)\\
& = r + \xi_{r,\phi}^{(n)},
\end{aligned}
\end{equation}
where $\xi_{r,\phi}^{(n)} = - R\cos(\phi-\psi_n) + \frac{R^2}{2r}\left(1-\cos^2(\phi-\psi_n)\right)$ denotes the difference between the propagation distance to the $n^{th}$ antenna and the propagation distance to the origin. The approximation (a) is derived from the second-order Taylor series expansion $\sqrt{1+x} = 1 + \frac{x}{2} - \frac{x^2}{8} + \mathcal{O}(x^3)$ assuming $r$ is large compared to the other terms\footnote{The second-order approximation is accurate enough when $r$ is larger than the Fresnel distance $r = \frac{D}{2}(\frac{D}{\lambda})^{1/3}$, which is often comparable to the array aperture and thus could be commonly satisfied~\cite{Janaswamy'17'm}.}. 

\begin{remark}
Similar to the near-field analysis for ULA systems~\cite{Cui'22'tcom}, when Taylor series expansion only keeps the first-order term, the propagation distance $r^{(n)}$ degenerates to the far-field condition where $r^{(n)} \approx r - R\cos(\phi-\psi_n)$. In this case, the near-field beam focusing vector~\eqref{eq: near beam vector} naturally degenerates into the far-field form defined in~\eqref{eq: far beam vector}. Therefore, the far-field beam steering vector is a special case of near-field beam focusing vector without higher-order Taylor series expansions.
\end{remark}

Different from far-field beam steering vectors which are only able to focus the signal power on specific angles, near-field beam focusing vectors could focus the power on specific angles and distances, i.e. specific locations in the whole two-dimensional (2D) space~\cite{zhang'22'j}. Therefore, apart from the angle information, the distance information of users also plays an indispensable role when performing beamforming. To this end, characterizing the beamforming property of UCA in the 2D space is crucial, which will be elaborated as follows.

\section{Analysis of UCA Near-Field Beamforming}\label{sec: beamform}
Since the property of near-field beamforming with UCA has fundamentally changed, in this section we aim to analyze its beamforming gain in both the angular and distance domains.
\vspace{-1em}
\subsection{Analysis of Beamforming Gain in the Angular Domain}\label{sec: ana angular}
The beamforming property in the angular domain is first discussed. Based on the widely adopted matched-filter beamforming methods, where the beamforming vectors are designed as the conjugate of near-field beam focusing vectors, the beamforming gain can be formulated as
\begin{equation}
\label{eq: gain near}
\begin{aligned}
&~~g(r_1, \phi_1, r_2, \phi_2) = |{\bf{b}}^H(r_1, \phi_1){\bf{b}}(r_2, \phi_2)|\\
& = \frac{1}{N} \left| \sum_{n=1}^{N} e^{j\frac{2\pi}{\lambda}(\sqrt{r_1^2+R^2-2r_1R\cos(\phi_1 - \psi_n)}-r_1)}\right. \\
& ~~~~~~~~ \left. \times e^{-j\frac{2\pi}{\lambda}(\sqrt{r_2^2+R^2-2r_2R\cos(\phi_2 - \psi_n)}-r_2)} \right|\\
&\approx \frac{1}{N} \left| \sum_{n=1}^{N} e^{-j\frac{2\pi}{\lambda}R(\cos(\phi_1-\psi_n) - \cos(\phi_2-\psi_n))} \right. \\
& ~~~~~~~~ \left.\times e^{j\frac{2\pi}{\lambda}R^2 \left( \frac{1-\cos^2(\phi_1-\psi_n)}{2r_1} - \frac{1-\cos^2(\phi_2-\psi_n)}{2r_2} \right) } \right|,
\end{aligned}
\end{equation}
where ${\bf{b}}(r_1, \phi_1)$ and ${\bf{b}}(r_2, \phi_2)$ are beam focusing vectors defined as in~\eqref{eq: near beam vector}. Since we first focus on the variation of the beamforming gain against radiating angles, a fixed propagation distance is assumed, i.e. $r_1=r_2=r$. Then, the property of beamforming gain can be described as follows.
\begin{lemma}
\label{lemma1}
The beamforming gain at location $(r, \phi_1)$ achieved by employing the near-field beamforming vector ${\bf{b}}(r, \phi_2)$ can be proved to approximate the far-field beamforming gain as
\begin{equation}
\label{eq: gain near angular 2}
\begin{aligned}
g(r, \phi_1, r, \phi_2) &= \left|{\bf{b}}^H(r, \phi_1) {\bf{b}}(r, \phi_2) \right|\\
& \approx \left|{\bf{a}}^H(\phi_1) {\bf{a}}(\phi_2) \right| \approx \left| J_0(\beta) \right|, 
\end{aligned}
\end{equation}
where $J_0(\cdot)$ denotes the zero-order Bessel function of the first kind and $\beta = \frac{4\pi R}{\lambda} \sin(\frac{\phi_2-\phi_1}{2})$.
\end{lemma}
\begin{proof}
The proof is provided in {\bf Appendix~\ref{app: lemma1}}.
\end{proof}
\par A figure of the absolute value of zero-order Bessel function of the first kind $|J_0(\beta)|$ is shown in Fig.~\ref{img: bessel}. It can be seen that the Bessel function achieves maximum when $\beta = 0$, which is equivalent to the accurate beamforming with $\phi_1 = \phi_2$. When $\beta>0$, $|J_0(\beta)|$ is decreasing with fluctuations as $\beta$ increases. Therefore, the acquirement of accurate angle information is essential to perform effective beamforming. Otherwise, the received signal power may dramatically decrease or even approach zero.
\begin{figure}[!t]
    \centering
    \setlength{\abovecaptionskip}{0.cm}
    \includegraphics[width=3in]{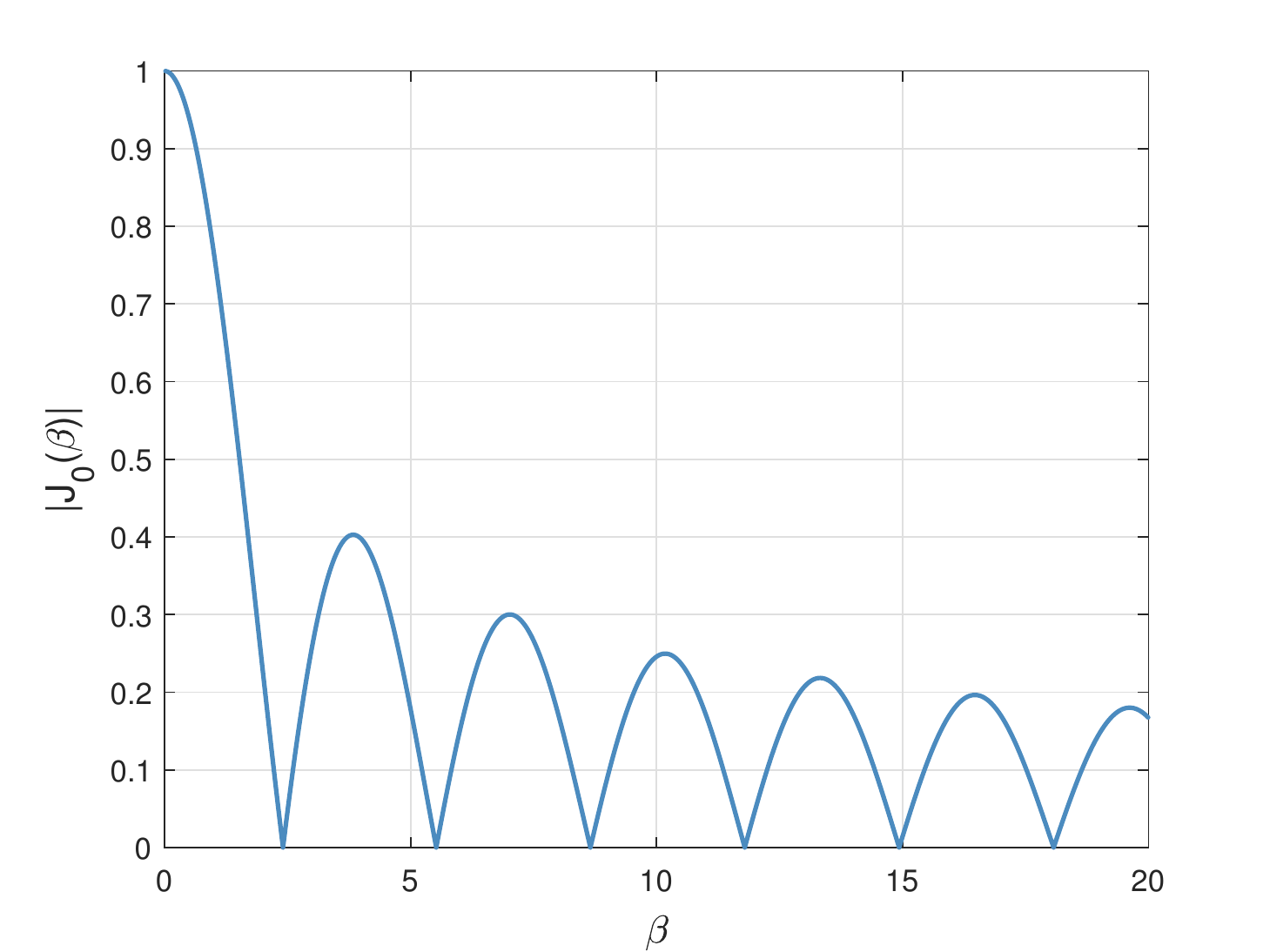}
    \caption{An example of the absolute value of the zero-order Bessel function of the first kind.}
    \label{img: bessel}
\end{figure}

\begin{remark}
The approximation is derived with the assumption that antennas in UCA are not separated too far, which is commonly adopted in existing research on UCA and could be satisfied with half-wavelength spacing~\cite{Zhang'17'tawp,Fanwei'19'tvt}. It is worth noting that only the first-order of Taylor series is kept in the proof, and thus the beamforming gain in the near-field approaches that in the far-field region. The reason for adopting the first-order approximation lies in that even when users are located in the near-field region, the first-order term still dominates when users are located at the same distance. This approximation could obtain good accuracy, which will be verified in simulations. Moreover, the second-order term of Taylor series will significantly affect the beamforming gain when the first-order term is eliminated, which will be shown in the following subsection.
\end{remark}
\vspace{-1em}
\subsection{Analysis of Beamforming Gain in the Distance Domain}\label{sec: ana distance}
What distinguishes near-field beamforming from classical far-field beamforming is that the beamforming gain varies with different distances. In this subsection, we focus on the beamforming performance in the distance domain, which means that $\phi_1 = \phi_2 = \phi$ is satisfied. Then the beamforming gain can be rewritten as
\begin{equation}
\label{eq: gain near distance}
\begin{aligned}
g(r_1, \phi, r_2, \phi) &= |{\bf{b}}^H(r_1, \phi){\bf{b}}(r_2, \phi)|\\
& \approx \frac{1}{N} \left|\sum_{n=1}^{N} e^{j \frac{2\pi}{\lambda}(\xi_{r_1,\phi}^{(n)} - \xi_{r_2,\phi}^{(n)})} \right|\\
& = \frac{1}{N} \left|\sum_{n=1}^{N} e^{j \frac{2\pi}{\lambda}\left\{\left(\frac{R^2}{2r_1}-\frac{R^2}{2r_2}\right)\left(1-\cos^2(\phi-\psi_n)\right)\right\} } \right|,
\end{aligned}
\end{equation}
where $\xi_{r_1,\phi}^{(n)}$ and $\xi_{r_2,\phi}^{(n)}$ are defined as~\eqref{eq: distance near}. The beamforming gain against different distances can be analytically approximated by the following lemma.
\begin{lemma}
\label{lemma2}
The beamforming gain of near-field beam focusing vectors with UCA can be approximated as follows
\begin{equation}
\label{eq: gain near distance 2}
\begin{aligned}
g(r_1, \phi, r_2, \phi) = \left|{\bf{b}}^H(r_1, \phi) {\bf{b}}(r_2, \phi) \right| \approx \left| J_0(\zeta) \right|, 
\end{aligned}
\end{equation}
where the variable $\zeta$ is defined as
\begin{equation}
\label{eq: gain near distance 3}
\begin{aligned}
\zeta = \frac{2\pi R^2}{\lambda}\left|\frac{1}{4r_1}-\frac{1}{4r_2}\right|. 
\end{aligned}
\end{equation}
\end{lemma}

\begin{proof}
The proof is provided in {\bf Appendix~\ref{app: lemma2}}.
\end{proof}
\par This lemma characterizes the focusing property of near-field beam focusing vectors in the distance domain. Unlike the far-field beamforming which is uniform with different transmission distances~\cite{Kallnichev'01}, the beamforming gain in the near-field region varies with the distance. For fixed $r_1$, the limit of the beamforming gain when $r_2$ approaches $0$ or $\infty$ could be expressed through the following corollary.
\begin{corollary}[Limit of Beamforming Gain]
\label{coro1}
For fixed $r_1$, the limit of beamforming gain when $r_2$ tends to 0 is expressed as
\begin{equation}
\label{eq: limit with small r}
\begin{aligned}
\lim_{r_2 \to 0+} g(r_1, \phi, r_2, \phi) = 0.
\end{aligned}
\end{equation}
Similarly, the limit of beamforming gain when $r_2$ tends to infinity can be written as
\begin{equation}
\label{eq: limit with large r}
\begin{aligned}
\lim_{r_2 \to \infty} g(r_1, \phi, r_2, \phi) = \left|J_0\left(\frac{\pi R^2}{2\lambda r_1}\right)\right|.
\end{aligned}
\end{equation}
\end{corollary}
\begin{proof}
According to~\eqref{eq: gain near distance 2}, $\zeta$ tends to infinity when positive $r_2$ tends to $0$. With $\lim_{x \to \infty} J_0(x) = 0$, equation~\eqref{eq: limit with small r} could be easily obtained\footnote{The limit derived in {\bf{Corollary}~\ref{coro1}} is built on a basis of the approximation in~\eqref{eq: distance near}. The approximation in~\eqref{eq: distance near} may not be accurate anymore for small $r_2$. Nevertheless, this limit analysis illustrates the variation trends for different $r_2$, indicating the necessity of accurate beamforming.}. Similarly, when $r_2$ tends to infinity, $\zeta = \frac{\pi R^2}{2\lambda r_1}$ could be obtained, which completes the proof of~\eqref{eq: limit with large r}.
\end{proof}

\par The corollary describes the limit of beamforming gain when the user is far away from the focal point of near-field beams. Generally speaking, according to the overall descending trend of $J_0(\beta)$, larger distance between the user and the focal point of beam will result in smaller beamforming gain, indicating the necessity of accurate beamforming in the distance domain.

\par In addition, we are also curious about the envelope of the beamforming gain, which indicates the variation trends of the beamforming without fluctuations.
\begin{corollary}[Upper Bound of Beamforming Gain]
\label{coro2}
The upper bound of beamforming gain in~\eqref{eq: gain near distance 3} for large $\zeta$ can be expressed 
\begin{equation}
\label{eq: asymp beam gain}
\begin{aligned}
g(r_1,\phi,r_2,\phi) \leq \frac{2}{\pi R}\sqrt{\frac{\lambda}{|1/r_1-1/r_2|}} = \frac{2}{\pi R}\sqrt{\frac{\lambda r_1 r_2}{|r_2-r_1|}}.
\end{aligned}
\end{equation}

\end{corollary}
\begin{proof}
According to the asymptotic property of $J_0(\cdot)$ that $J_0(x)=\sqrt{\frac{2}{\pi x}} \cos\Big(x-\frac{\pi}{4} + O\left(x^{-3/2}\Big)\right)$, $\left|J_0(\zeta)\right| \leq \sqrt{\frac{2}{\pi \zeta}}$ could be ensured for large $\zeta$, completing the proof.
\end{proof}
\par This corollary characterizes the upper bound as well as the convergence speed of beamforming gain in the distance domain. It is shown that beamforming gain decreases as $R$ or distance differences $\left|\frac{1}{r_1}-\frac{1}{r_2}\right|$ scale up. For fixed $r_1$ and $r_2$, the convergence speed of the beamforming gain is of the rate $\frac{1}{R}$. This asymptotic analysis will be verified in Section~\ref{sec: sim}.

\par Apart from the limit analysis, we are also concerned about the property of beamforming gain around the focal point. For far-field beamforming, the concept of beamwidth is used to describe the main lobe range of beams in the angular domain. Similarly, the $3~\rm{dB}$ \emph{depth-of-focus} of the near-field beam could also be defined to capture the main lobe of the beam in the distance domain with the following corollary as in~\cite{emil'21}.

\begin{corollary}[Depth-of-Focus of UCA]
\label{coro3}
The $3~\rm{dB}$ range of the near-field beamforming gain in the distance domain with UCA can be expressed by
\begin{equation}
\label{eq: beam depth}
{\rm{BD}_{3dB}} =\left\{
\begin{aligned}
\frac{4\pi \eta \lambda R^2r_0^2}{\pi^2 R^4-4\eta^2 \lambda^2 r_0^2}, \quad r_0 &< \frac{\pi R^2}{2\eta \lambda}\\
\infty~~~~~~~~, \quad r_0 &\geq \frac{\pi R^2}{2\eta \lambda},\\
\end{aligned}
\right
.
\end{equation}
where $\eta$ represents the $3~\rm{dB}$ threshold defined as $J_0(\eta) = 1/2$ and $\eta = 1.521$ is numerically calculated from Bessel functions.
\end{corollary}
\begin{proof}
We assume that the focal point of near-field beamforming is $(r_0, \phi)$. According to {\bf{Lemma~\ref{lemma2}}}, the $3~\rm{dB}$ edge of the beam needs to satisfy $|J_0(\xi_{\rm{edge}})| = 1/2$ where $\xi_{\rm{edge}} = \frac{2\pi R^2}{\lambda}\left|\frac{1}{4r_0}-\frac{1}{4r}\right|$. With the assumption $J_0(\eta) = 1/2$, we can obtain
\begin{equation}
\label{eq: app3 1}
\begin{aligned}
\left|\frac{1}{4r_0}-\frac{1}{4r}\right| = \frac{\lambda \eta}{2\pi R^2}.
\end{aligned}
\end{equation}
With the constraint $r>0$, we can only have one-side edge when $r_0 \geq \frac{\pi R^2}{2\eta \lambda}$, which results in an infinite depth-of-focus. Otherwise, the depth-of-focus can be expressed as $r_{\rm{max}}-r_{\rm{min}}$, where $r_{\rm{max}}$ and $r_{\rm{min}}$ are the two solutions of~\eqref{eq: app3 1} as
\begin{equation}
\label{eq: app3 2}
\begin{aligned}
r = \frac{\pi R^2 r_0}{\pi R^2 \mp 2\eta \lambda r_0}.
\end{aligned}
\end{equation}
\par In addition, owing to the property of the zero-order Bessel function, there only exists one solution satisfying $J_0(\eta) = 1/2$ and $J_0(\cdot)$ is monotonically decreasing in the range $[0, \eta]$, which ensures that the beamforming gain always exceeds $1/2$ in the whole $3~\rm{dB}$ range, which completes the proof.
\end{proof}

\begin{corollary}[Asymptotic Property of Depth-of-Focus]
\label{coro4}
If the half-wavelength spacing is adopted for UCA, the depth-of-focus of UCA tends to zero when the number of antennas tends to infinity, which is to say
\begin{equation}
\label{eq: asymptotic}
\begin{aligned}
\lim_{N \to \infty} {\rm{BD}_{3dB}} = 0.
\end{aligned}
\end{equation}
\end{corollary}
\begin{proof}
With the half-wavelength spacing assumption, the UCA radius $R$ tends to infinity as the number of antennas tends to infinity. Then, the condition $r_0 < \frac{\pi R^2}{2\eta \lambda}$ always holds for arbitrarily fixed $r_0$. Therefore, the $3~\rm{dB}$ depth-of-focus could be rewritten as $\frac{4\pi \eta \lambda r_0^2}{\pi^2 R^2-4\eta^2 \lambda^2 r_0^2/R^2}$, indicating that the $3~\rm{dB}$ depth-of-focus converges to zero when $R \to \infty$, which completes the proof.
\end{proof}
\begin{remark}
{\bf{Corollary~\ref{coro3}}} and {\bf{Corollary~\ref{coro4}}} characterize the asymptotic focusing property of a large circular array. According to equation~\eqref{eq: beam depth}, the depth-of-focus monotonically decreases as the radius $R$ scales up, indicating a stronger focusing ability in the distance domain could be obtained if the array aperture is enlarged. The enhanced focusing ability of UCA could be leveraged to mitigate the power leakage in undesired locations, enhancing the spectrum efficiency especially in LoS dominating scenarios, which will be discussed in the next subsection.
\end{remark}

Another important issue of UCA beamforming is to determine the region where the spherical-wave propagation model needs to be adopted. Rayleigh distance (or Fraunhofer boundary) has been widely adopted to partition the far-field and near-field regions, which is defined with the criterion of phase discrepancy lower than $\pi/8$ when planar-wave approximations are adopted~\cite{Sherman'62'j}. Nevertheless, it was proved that the widely used Rayleigh distance is not accurate enough to characterize the beamforming performance, which is more critical for communications~\cite{cui'21}. To solve this problem, effective Rayleigh distance (ERD) is introduced to mark the region where beamforming loss adopting classical far-field steering vectors is larger than the threshold, which is a more accurate near-field region for communications~\cite{Cui'22'tcom}. According to the definition, the ERD can be formulated as
\begin{equation}
\label{eq: erd}
\begin{aligned}
r_{\rm{ERD}}(\phi) \mathop{=}\limits^{\Delta} {\underset{r}{{\arg\max}} \, \left\{ 1- |{\bf{b}}^H(r,\phi){\bf{a}}(\phi)| \geq \delta \right\} },
\end{aligned}
\end{equation}
where $\delta$ is a predefined beamforming loss threshold and $1 - |{\bf{b}}^H(r,\phi){\bf{a}}(\phi)|$ characterizes the loss of beamforming gain with far-field beamforming. Then, the ERD for UCA can be provided through the following lemma.
\begin{lemma}
\label{lemma3}
The ERD defined in~\eqref{eq: erd} for UCA can be expressed as
\begin{equation}
\label{eq: erd uca}
\begin{aligned}
r_{\rm{ERD}}^{({\rm{C}})}(\phi) = \frac{\pi R^2}{2\lambda J_0^{-1}(1-\delta)} = \epsilon_{\rm{C}} \frac{2D_{\rm{C}}^2}{\lambda},
\end{aligned}
\end{equation}
where the notation $J_0^{-1}(\cdot)$ is defined as the inverse function of Bessel function in the main lobe, which is to say $J_0^{-1}(x) = \left\{ {\underset{y_0}{{\arg\min}} \, |J_0(y_0)| = x} \right\}$ for $x \in [0,1]$. The superscript $({\rm{C}})$ denotes UCA and $\epsilon_{\rm{C}} = \frac{\pi}{16 J_0^{-1}(1-\delta)}$. 
\end{lemma}
\begin{proof}
According to {\bf {Lemma}~\ref{lemma2}}, the beamforming gain adopting the far-field beam steering vector ${\bf{a}}(\phi)$ at the location $(r, \phi)$ can be formulated as
\begin{equation}
\label{eq: lemma 3 eq1}
\begin{aligned}
g = \left|{\bf{b}}^H(r, \phi){\bf{a}}(\phi) \right| = \left|{\bf{b}}^H(r, \phi){\bf{b}}(\infty, \phi) \right| = \left|J_0 \left(\frac{\pi R^2}{2\lambda r} \right) \right|.
\end{aligned}
\end{equation}
According to the definition of $J_0^{-1}(\cdot)$, $1 - g \leq \delta$ always holds for $\frac{\pi R^2}{2\lambda r} \leq J_0^{-1}(1-\delta)$. Therefore, $r_{\rm{ERD}}^{({\rm{C}})} = \frac{\pi R^2}{2\lambda J_0^{-1}(1-\delta)}$ can be obtained. Besides, it can be seen from~\eqref{eq: lemma 3 eq1} that the ERD is invariant for different angles, therefore $\phi$ could be omitted. This completes the proof.
\end{proof}
\par The ERD of UCA is invariant for different spatial angles, which is different from the ULA scenario where the ERD significantly reduces for large spatial directions~\cite{Cui'22'tcom}. More comparisons between UCA and ULA are provided in the following subsection.
\vspace{-1em}
\subsection{Comparison with the Uniform Linear Array}\label{sec: discussion}
\par We assume that both ULA and UCA share the same array aperture\footnote{In practical implementation, the BS usually has a restrict constraint on the largest array aperture, instead of the number of antennas equipped.}, that is to say, $D_{\rm{L}} = D_{\rm{C}} = 2R$. According to the results in~\cite{Cui'22'tcom}, the ERD for ULA can be expressed as
\begin{equation}
\label{eq: erd ULA}
\begin{aligned}
r_{\rm{ERD}}^{({\rm{L}})}(\phi) = \epsilon_{\rm{L}} \frac{2D_{\rm{L}}^2 \cos^2\phi}{\lambda},
\end{aligned}
\end{equation}
which is related to the direction $\phi$. The threshold $\epsilon_{\rm{L}}$ is determined according to the predefined beamforming loss of $\delta$. For instance, $\epsilon_{\rm{L}} = 0.367$ can be ensured for $\delta = 0.05$. In addition, it can be shown that ERD significantly reduces as the direction of user $\phi$ increases. The reason is that the effective array aperture $D_{\rm{L}}\cos\phi$ reduces for large angles, resulting in smaller phase differences between far-field planar-wave and near-field spherical-wave models. On the contrary, due to the rotational symmetry of UCA, the effective array aperture remains the same in different directions, which results in invariant ERD as in~\eqref{eq: erd uca}. This feature indicates the possibility of UCA to provide uniform and enlarged near-field regions compared with ULA. To compare the near-field regions, the ratio of ERD is defined as
\begin{equation}
\label{eq: ratio erd}
\begin{aligned}
\rho(\phi) = \frac{r_{\rm{ERD}}^{({\rm{L}})}(\phi)}{r_{\rm{ERD}}^{({\rm{C}})}(\phi)} = \frac{\epsilon_{\rm{L}} \cos^2\phi}{\epsilon_{\rm{C}}}.
\end{aligned}
\end{equation}
\par The shape of the ERD as well as the ratio with threshold $\delta = 0.05$ is plotted in Fig.~\ref{img: UCA ULA}, where solid lines denote the ERD calculated according to the definition in~\eqref{eq: erd} and dashed lines denote the estimated ERD obtained from~\eqref{eq: erd uca} and~\eqref{eq: erd ULA}. It can be seen that ERD of UCA exceeds that of ULA over all angles with the same array aperture. Therefore, it shows that the near-field region measured by ERD could be enlarged with UCA. Moreover, the estimations of ERD in equations (21) and (23) are consistent with the accurate calculations, showing high estimation accuracy. 
\begin{figure}[!t]
    \centering
    \subfigure[Edge of ERD]{
    	\includegraphics[width=1.48in]{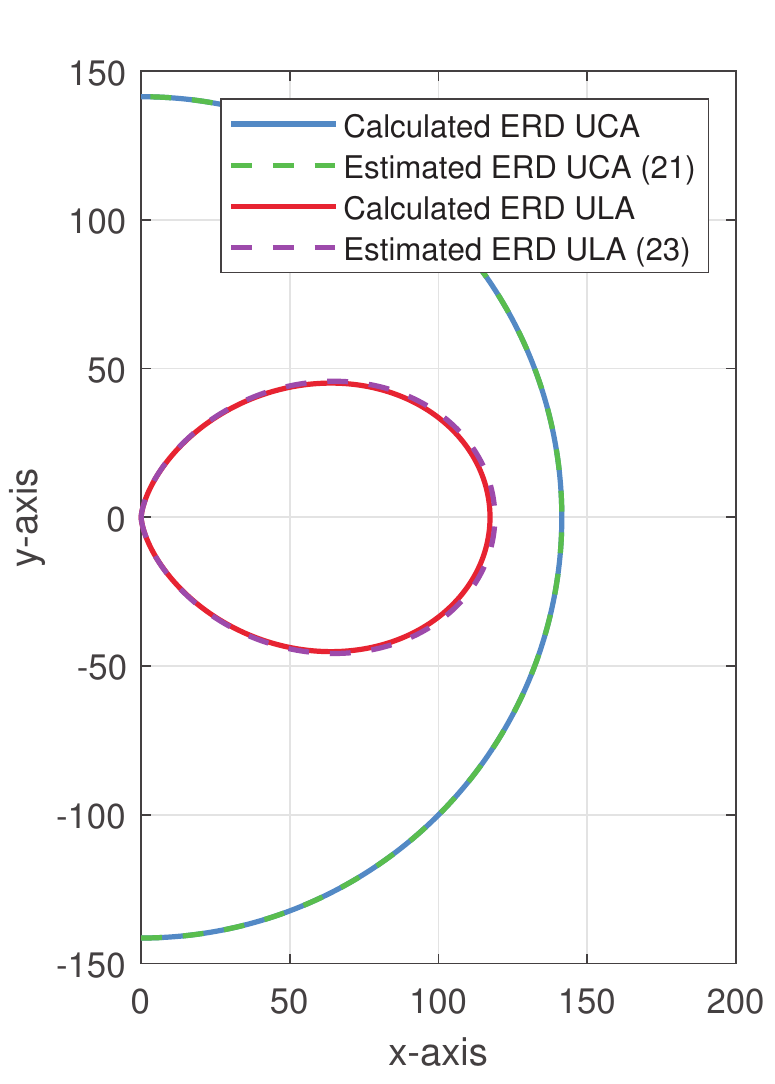}}
    \subfigure[Ratio of ERD]{
    	\includegraphics[width=1.5in]{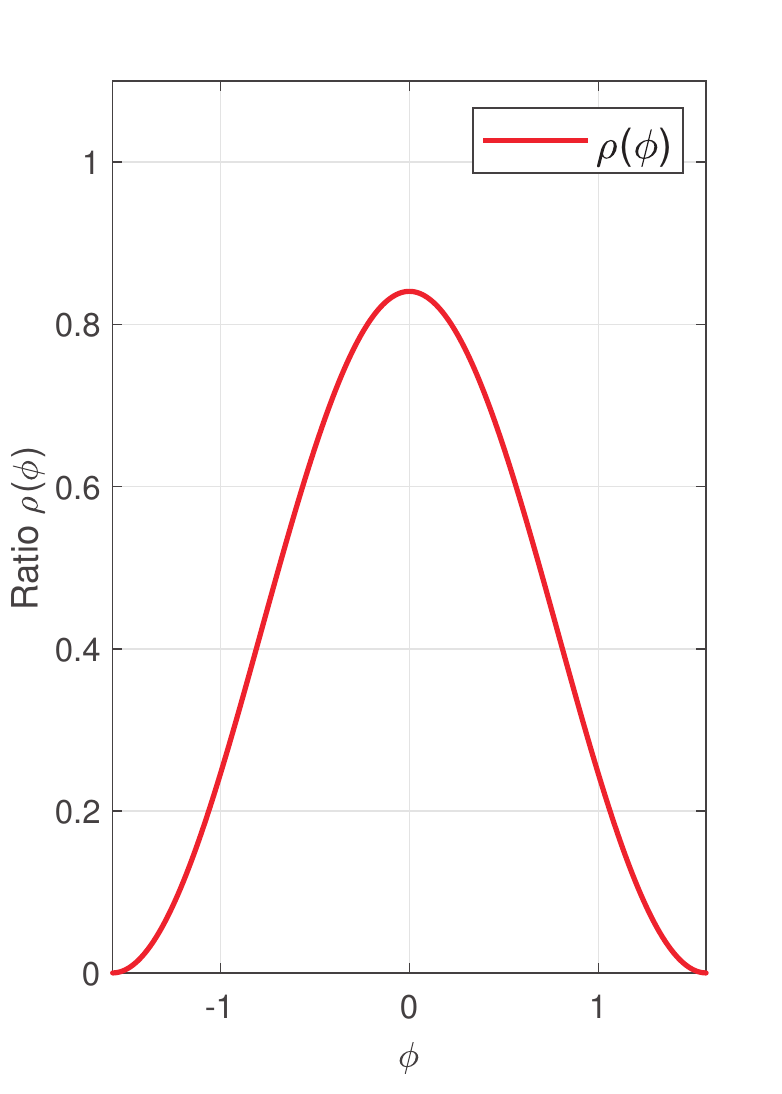}}
    \caption{Comparison of ERD for UCA and ULA over different angles. The center of UCA and ULA are both located at the origin. The ULA is placed along the $y$-axis.}
    \label{img: UCA ULA}
\end{figure}
\par In addition, the beamforming gain in the distance domain for ULA and UCA has different features. In massive MIMO systems for 5G communication, multiple users located in different angles could be served simultaneously with SDMA to significantly enhance spectrum efficiency. The classical SDMA scheme is realized by the angular focusing ability, i.e. focusing the signal energy in desired angles to increase received signal power and mitigate undesired interferences. Unlike far-field communications which are only capable of focusing signal energy on specific directions, the signal energy could be concentrated on specific directions and distances, i.e. locations, in the near field. The focusing property of near-field beams could be employed to enable transmissions with users not only in different angles but also in different distances~\cite{Wu'22'jsac}.
\par To show that compared with ULA, UCA possesses stronger focusing ability in distance domain, the comparison of beamforming gain for ULA and UCA is shown in Fig.~\ref{img: UCA ULA gain}. The beamforming gain is derived by employing far-field beam steering vectors in near field with $g = |{\bf{b}}^H(r,\phi){\bf{a}}(\phi)|$. The array aperture is set to $1.27\, {\rm{m}}$ with half-wavelength spacing for UCA and ULA working at 30 GHz. It can be seen that, compared with the smooth decreasing trend of beamforming gain with ULA, the beamforming gain decreases faster with UCA, indicating that UCA is capable of focusing signal power in smaller range and mitigating power leakage. Note that the results are consistent with the conclusion in Fig.~\ref{img: UCA ULA} that a larger near-field region could be obtained with UCA. 
\par Moreover, owing to the property of Bessel functions, there exist several zeros of beamforming gain, indicating that zero beamforming gain can be obtained even when the beams are aligned along the right direction. This phenomenon indicates a possibility for further spectrum efficiency enhancement in multi-user MIMO systems. Specifically, if ULA is adopted and two users are randomly aligned in the same direction, near-field beams focusing on either of the users inevitably generates interferences at the other user. However, if UCA is employed, one user can be located around the zero point of the beam which focuses the energy on the location of the other user. In such cases, interference-free transmissions could be obtained, which paves a way to utilize spatial resources in the distance domain to enhance spectrum efficiency.

\begin{figure}[!t]
    \centering
    \setlength{\abovecaptionskip}{0.cm}
    \includegraphics[width=3in]{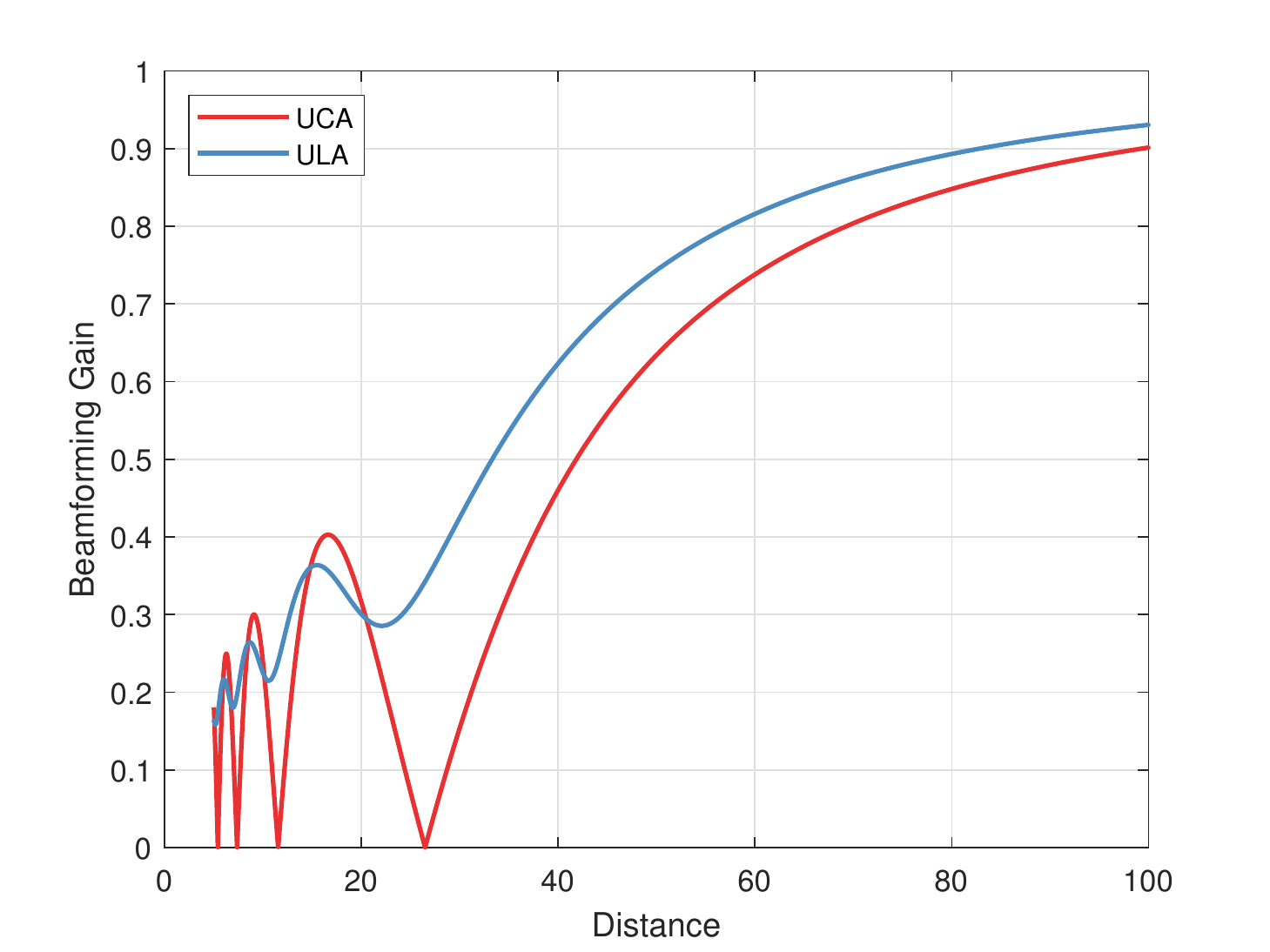}
    \caption{Comparison of the beamforming gain of UCA and ULA.}
    \label{img: UCA ULA gain}
\end{figure}
\vspace{-1em}
\section{Design of the Concentric-Ring Codebook}\label{sec: codebook}
\par In current 5G systems, a beam codebook perfectly matching the far-field channel is usually defined in advance, in order to perform efficient beamforming, beam training, channel feedback, etc~\cite{Heath'16'j}. In this section, the near-field concentric-ring codebook for UCA is introduced to adapt to the emerging near-field propagation model with the widely adopted beamsteering design methods, which employ beam steering (or focusing) vectors to construct the codebook~\cite{Heath'15'j}. 
\par We are aiming to design the near-field UCA codebook ${\bf{W}} = \left\{ {\bf{b}}(r_1, \phi_1),\cdots, {\bf{b}}(r_{N_c}, \phi_{N_c}) \right\}$, where $N_c$ denotes the number of codeword vectors in the codebook. For fixed codebook size, the common design criterion is to minimize the correlation of codeword vectors, which is to say
\begin{equation}
\label{eq: codebook}
\begin{aligned}
\underset{1\leq i < j \leq N_c}{{\min}} ~ \left|{\bf{b}}^H(r_i, \phi_i){\bf{b}}(r_j, \phi_j) \right|.
\end{aligned}
\end{equation}
Or in another way, we can maximize the codebook size $N_c$ with the constraint on correlation of codewords as $\left|{\bf{b}}^H(r_i, \phi_i){\bf{b}}(r_j, \phi_j) \right| \leq \Delta$ for $i \neq j$~\cite{Cui'22'tcom}, where $\Delta$ denotes the threshold on the maximum allowable correlation.

\par To solve this problem, a sampling procedure determining the values of $r$ and $\phi$ needs to be carried out. According to {\bf{Lemma~\ref{lemma1}}}, the beamforming gain in the angular domain is independent of the propagation distance for users at the same distances. Similarly, the beamforming gain in the distance domain is also independent of angles according to {\bf{Lemma~\ref{lemma2}}}. Therefore, the angular sampling and distance sampling could be decoupled and performed separately, with the target of performing dense sampling satisfying the correlation constraint. Since the sampled distances remain the same at any angle, the samples form the shape of multiple concentric rings, which is termed the concentric-ring codebook. The design procedure of the proposed concentric-ring codebook includes two stages.
\par First, we focus on the sampling in the angular domain. When neglecting the influence of propagation distance $r$, the angles could be selected with the correlation lower or equal to the predetermined threshold $\Delta$, which is to say $g(r,\phi_1,r,\phi_2) = \left| J_0(\beta) \right| \leq \Delta$. Since $|J_0(\cdot)|$ is monotonically decreasing in its main lobe, we can ensure the correlation of neighboring samples is equal to the threshold $\Delta$ by $\phi_{\Delta} = 2\sin^{-1}\left(\frac{\lambda J_0^{-1}(\Delta)}{4\pi R}\right)$ for $J_0^{-1}(\cdot)$ defined in equation~\eqref{eq: erd uca}\footnote{Such sampling method can only ensure that the correlation of neighboring samples is less than $\Delta$, since only the main lobe is considered to determine the space of neighboring samples. Owing to the oscillation of Bessel functions, the correlation of non-neighboring samples may exceed the threshold. To provide a succinct design method, we exclude such scenarios by assuming $\Delta \geq 0.403$, indicating that allowable threshold is greater than peak of the first sidelobe.}. Thus, we can obtain the angular samples as
\begin{equation}
\label{eq: angular samples}
\begin{aligned}
\phi_{s_1} =  s_1 \phi_{\Delta},~~{\rm{for}}~s_1 = 0,1,2,\cdots,S_1,
\end{aligned}
\end{equation}
where $S_1 = \left \lfloor \frac{2\pi}{\phi_{\Delta}} \right \rfloor - 1$ which ensures that the samples are in the primary period $\phi_{s_1} \in [0, 2\pi)$.
\par Second, once the sampled angles are determined, we focus on distance sampling at selected angles. According to~{\bf{Lemma}~\ref{lemma2}}, the correlation of distance samples in the same direction is lower than the threshold by $\left|\frac{1}{r_p} - \frac{1}{r_q}\right| \geq \frac{\lambda J_0^{-1}(\Delta)}{2\pi R^2}$. Therefore, we could follow a sampling method as
\begin{equation}
\label{eq: distance samples}
\begin{aligned}
r_{s_2} = \frac{1}{s_2} \frac{2\pi R^2}{\lambda J_0^{-1}(\Delta)},~~{\rm{for}}~s_2 = 0,1,\cdots,S_2.
\end{aligned}
\end{equation}
Since it is very rare for communications to happen at very near locations and therefore a minimum distance $r_{\rm{min}}$ is assumed. The parameter $S_2 = \left \lfloor \frac{2\pi R^2}{\lambda J_0^{-1}(\Delta) r_{\rm{min}}} \right \rfloor$ in~\eqref{eq: distance samples} aims to fulfill the minimum distance constraint and $s_2 =0$ represents the sampling at infinite distance, which is equivalent to classical far-field beamforming. The distance sampling method above naturally satisfies the correlation constraints.

\par Finally, since the procedure determining $\phi_{s_1}$ in angle domain and the procedure determining $r_{s_2}$ in distance domain are independent of each other, these two processes could be performed respectively and together constitute the whole codebook as $(r_{s_2}, \phi_{s_1})$. The construction process is summarized in {\bf{Algorithm}~\ref{alg:1}}. The algorithm complexity is mainly determined by computing the required phase shifts for each codeword, which can be expressed as $\mathcal{O}(S_a S_d N)$, where $S_a = \left \lfloor \frac{2\pi}{\phi_{\Delta}} \right \rfloor + 1$ and $S_d = \left \lfloor \frac{r_{\Delta}}{r_{\rm{min}}} \right \rfloor + 1$ denote the number of samples in angular and distance domain, respectively. It is noteworthy that although the increased near-field codebook size may introduce the problem of overwhelming overheads in UCA systems, several works proposed to employ deep learning or hierarchical methods to reduce overheads in ULA systems~\cite{pancunhua'23'cl}, which provided potential solutions to tackle the overheads problem with UCA.

\begin{figure}[!t]
    \centering
    \setlength{\abovecaptionskip}{0.cm}
    \includegraphics[width=2.2in]{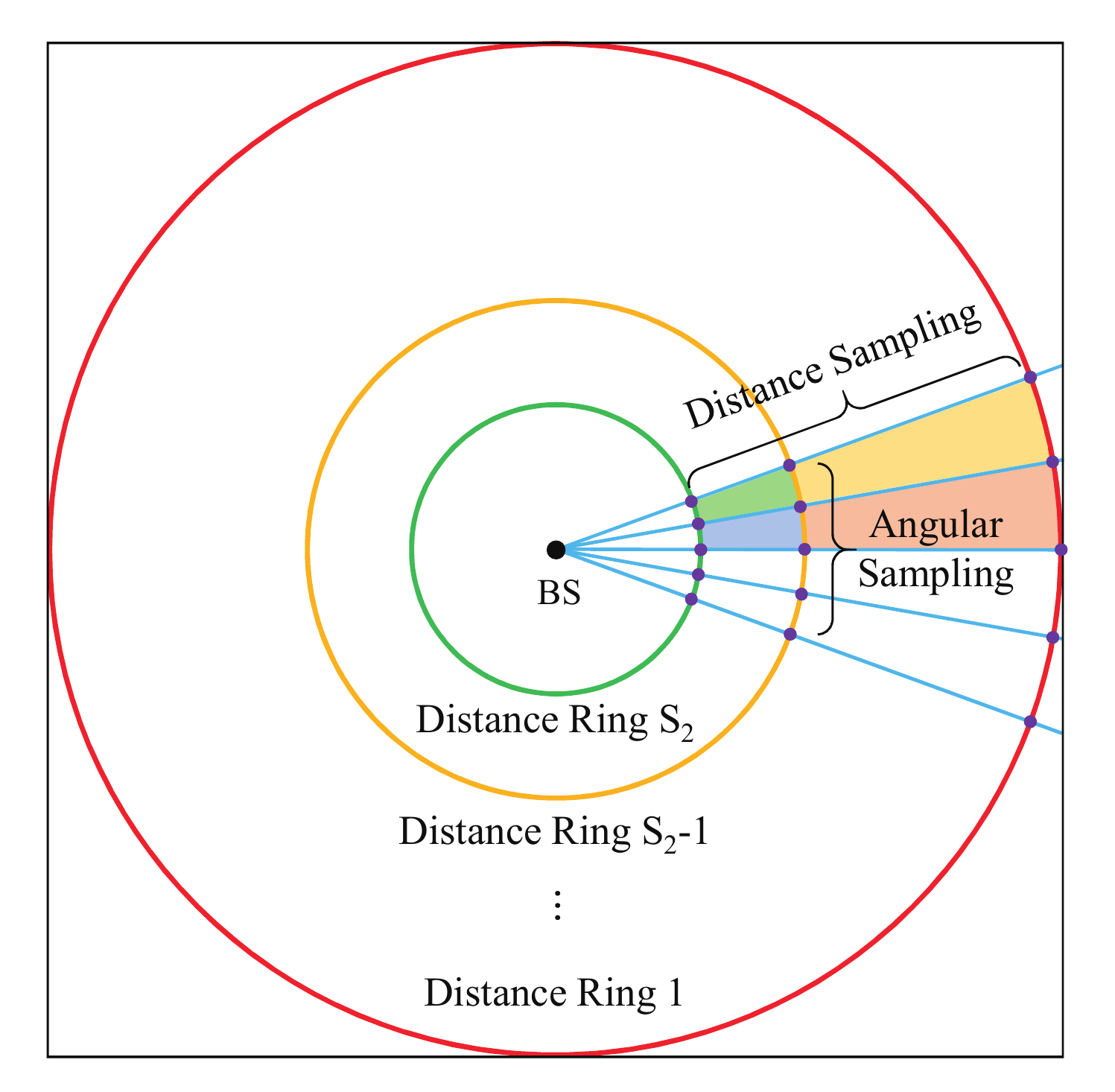}
    \caption{An example of the concentric-ring codebook for UCA.}
    \label{img: codebook}
\end{figure}

\begin{algorithm}[!t] 
    \caption{Construction of Concentric-Ring Codebook ${\bf{W}}$.} 
    \label{alg:1} 
    \begin{algorithmic}[1] 
        \REQUIRE ~ 
        Minimum distance $r_{\rm{min}}$; threshold $\Delta$; radius of the circular array $R$; wavelength $\lambda$;
        \ENSURE ~ 
        Concentric-ring codebook ${\bf{W}}$
        \STATE $s_1 = 0$, $s_2 = 0$, ${\bar{\Delta}} = J_0^{-1}(\Delta)$;
        \STATE $\phi_{\Delta} = 2\sin^{-1}(\frac{\lambda {\bar{\Delta}}}{4\pi R})$ by~\eqref{eq: gain near angular 2}, $r_{\Delta} = \frac{2\pi R^2}{\lambda{\bar{\Delta}}}$ by~\eqref{eq: gain near distance 2};
        \STATE $\Xi_{\phi} = \{0,1,\cdots, \left \lfloor \frac{2\pi}{\phi_{\Delta}} \right \rfloor \}$, $\Xi_{r} = \{0,1,\cdots, \left \lfloor \frac{r_{\Delta}}{r_{\rm{min}}} \right \rfloor \}$;
        \FOR{$s_1 \in \Xi_{\phi}$, $s_2 \in \Xi_{r}$}
        \STATE $\phi_{s_1} = s_1 \phi_{\Delta}$, $r_{s_2} = \frac{1}{s_2} r_{\Delta}$;
        \ENDFOR
        \STATE ${\bf{W}} = \left\{ {\bf{b}}(r_{s_2}, \phi_{s_1}) | s_1 \in \Xi_{\phi}, s_2 \in \Xi_{r} \right\}$
        \RETURN ${\bf{W}}$
    \end{algorithmic}
\end{algorithm}
\par An example of the concentric-ring codebook is shown in Fig.~\ref{img: codebook}. The concentric rings represent the sampling in distance domain, the radius of which is determined according to the distance sampling rule in~\eqref{eq: distance samples}. The rays radiated from BS represent the sampled angles according to~\eqref{eq: angular samples}. Finally, the intersections of rings and rays are sampled points that the near-field beamforming vectors should focus on.
\vspace{-1em}
\section{Generalization to Cylindrical Array}\label{sec: cyl}
\par For practical purposes, UPA is employed more commonly compared with ULA to achieve efficient utilization of the space in BS. Similarly, the 2D UCA could be generalized into the 3D structure array, cylindrical array, to ensure a coverage region in the azimuth plane as well as the elevation plane, whose array geometry is shown in Fig.~\ref{img: cylindrical}~\cite{Hussain'2005}. In this paper, we assume that all antennas could contribute to beamforming for illustration simplicity, which is a widely adopted assumption in theoretical analysis of cylindrical arrays~\cite{Hussain'2005,tanke'21'icsp}. Thus, the cylindrical array could be viewed as an aggregation of $2M+1$ concentric UCAs which are uniformly positioned along the $z$-axis with $z=md$, where $m=0, \pm1, \cdots, \pm M$ and the spacing of UCAs is $d$. To ensure the symmetry of the structure, the origin is set as the center of the middle UCA. Each UCA contains $N$ uniformly distributed elements as in Fig~\ref{img: UCA strcture}. Thus, the cylindrical array consists of $N\bar{M}$ antenna elements in total, where ${\bar{M}}$ is defined as ${\bar{M}} = 2M+1$. The user is considered to be located in the 3D space indexed by $\Theta = (r, \theta, \phi)$ where $r$, $\theta$, and $\phi$ denote the distance from the origin, the elevation angle, and the azimuth angle, respectively.

\begin{figure}[!t]
    \centering
    \setlength{\abovecaptionskip}{0.cm}
    \includegraphics[width=3.5in]{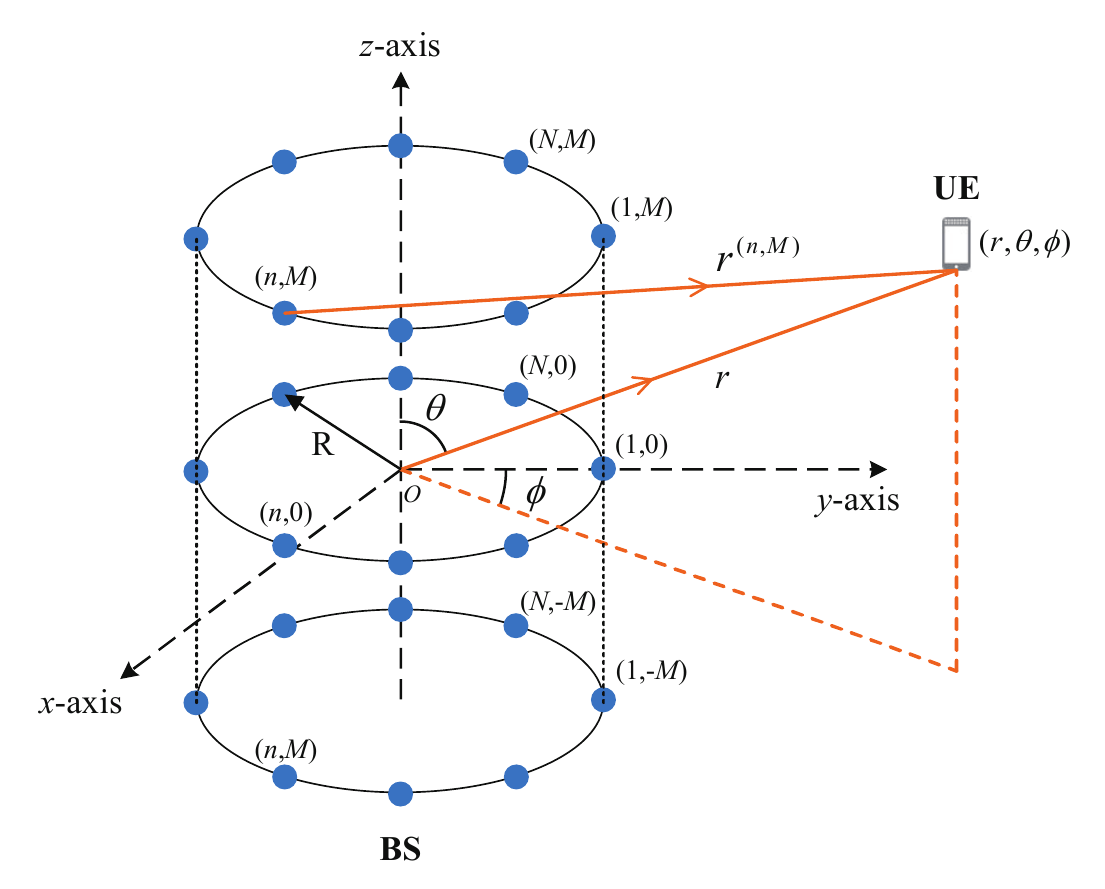}
    \caption{Geometry of the cylindrical array in 3D space.}
    \label{img: cylindrical}
    \vspace{-1.5em}
\end{figure}

\begin{figure*}[b]
\hrulefill
\begin{equation}
\label{eq: cylin distance}
\begin{aligned}
r^{(m,n)} &= \sqrt{r^2 + R^2 -2rR\sin\theta \cos(\phi-\psi_n) - 2mdr\cos\theta+m^2d^2}\\
&\approx r \underbrace{- R\sin\theta\cos(\phi-\psi_n)-md\cos\theta}_{\rm{first\mbox{-}order\ approximation}}\\
&~~~~\underbrace{ +\frac{R^2}{2r}(1-\sin^2\theta\cos^2(\phi-\psi_n))+\frac{m^2d^2}{2r}(1-\cos^2\theta)-\frac{Rd}{r}m\sin\theta\cos\theta\cos(\phi-\psi_n),}_{\rm{second\mbox{-}order\ approximation}}
\end{aligned}
\end{equation}
\end{figure*}

\par According to the geometrical relationship, the distance between user and the $(m,n)^{\rm{th}}$ antenna could be expressed as~\eqref{eq: cylin distance} at the bottom of next page, where the approximation is derived from the second-order Taylor series again. Note that if only the first-order approximation is adopted, we could obtain the classical far-field beamforming vector as in~\cite{Hussain'2005}. Following a similar process in Section~\ref{sec: sys}, the near-field beamforming vector can be defined as
\begin{equation}
\label{eq: cylin near-field vec}
\begin{aligned}
{\bf{b}}_{\rm{C}}(\Theta) &= \frac{1}{\sqrt{N {\bar{M}}}}\left[e^{-j\frac{2\pi}{\lambda}(r^{(1,-M)}-r)},\cdots, e^{-j\frac{2\pi}{\lambda}(r^{(N,M)}-r)} \right]^T,
\end{aligned}
\end{equation}
where the subscript $({\rm{C}})$ denotes the cylindrical array. Then, the beamforming gain is written as
\begin{equation}
\label{eq: gain cylin}
\begin{aligned}
g_{\rm{C}}(\Theta_1, \Theta_2) &= \frac{1}{N {\bar{M}}} \left|{\bf{b}}_{\rm{C}}^H(\Theta_1){\bf{b}}_{\rm{C}}(\Theta_2)\right|.
\end{aligned}
\end{equation}
It is worth noting that the expression in~\eqref{eq: gain cylin} is a universal formulation applicable to any locations of users and any settings of cylindrical array.
\par Apart from the universal expression above, to gain more intuition of the beamforming property of cylindrical arrays as well as highlighting its characteristics different from far-field beamforming, we now focus on the beamforming gain only in the distance domain where $\theta_1 = \theta_2$ and $\phi_1 = \phi_2$. In this case, equation~\eqref{eq: gain cylin} could be rewritten as
\begin{equation}
\label{eq: gain cylin distance}
\begin{aligned}
g_{\rm{C}}({\bar{\Theta}}_1, {\bar{\Theta}}_2) = \frac{1}{N {\bar{M}}} \left| \sum_{m=-M}^{M} \sum_{n=1}^{N} e^{j\frac{2\pi}{\lambda}(\frac{1}{r_1}-\frac{1}{r_2})(\chi_1 + \chi_2 + \chi_3)} \right|,
\end{aligned}
\end{equation}
where ${\bar{\Theta}}_1 = (r_1, \theta, \phi)$, ${\bar{\Theta}}_2 = (r_2, \theta, \phi)$ and the variables are defined as
\begin{equation}
\label{eq: chi define}
\begin{aligned}
\chi_1 &= \frac{R^2}{2}\left(1-\sin^2\theta\cos^2(\phi-\psi_n) \right)\\
\chi_2 &= \frac{m^2d^2}{2}(1-\cos^2\theta)\\
\chi_3 &= Rdm\sin\theta\cos\theta\cos(\phi-\psi_n).
\end{aligned}
\end{equation}
Due to the difficulties in addressing those two summations, we further focus on the beamforming performance on the $z=0$ plane, i.e. assuming $\theta = \pi/2$. The beamforming gain could be characterized through the following lemma.
\begin{lemma}
\label{lemma4}
The beamforming gain at location ${\bar{\Theta}}_1 = (r_1, \pi/2, \phi)$ achieved by employing near-field beam focusing vector ${\bf{b}}_{\rm{C}}({\bar{\Theta}}_2)$ with ${\bar{\Theta}}_2 = (r_2, \pi/2, \phi)$ is approximated by
\begin{equation}
\label{eq: gain cylin distance 2}
\begin{aligned}
g_{\rm{C}}({\bar{\Theta}}_1, {\bar{\Theta}}_2) \approx \left|G(\mu)J_0(\zeta)\right|, 
\end{aligned}
\end{equation}
where functions are defined as $G(\mu) = \frac{C(\mu)+jS(\mu)}{\mu}$, $C(\mu) = \int_{0}^{\mu} \cos(\frac{\pi}{2}t^2) {\rm{d}}t$, and $S(\mu) = \int_{0}^{\mu} \sin(\frac{\pi}{2}t^2){\rm{d}}t$~\cite{Sherman'62'j}, with $\mu = \sqrt{\frac{2M^2d^2}{\lambda}\left|\frac{1}{r_1}-\frac{1}{r_2}\right|}$ and $\zeta = \frac{2\pi R^2}{\lambda}\left|\frac{1}{4r_1}-\frac{1}{4r_2}\right|$.
\end{lemma}
\begin{proof}
The proof is provided in {\bf Appendix~\ref{app: lemma4}}.
\end{proof}

\par Compared with UCA, the expression of beamforming with cylindrical arrays is more complicated. Since the cylindrical array could be viewed as a series of linearly spaced UCAs, the expression of beamforming gain in~\eqref{eq: gain cylin distance 2} could be decomposed into two components, the UCA component represented by $J_0(\zeta)$, and the ULA component represented by $G(\mu)$. Therefore, if we only increase $M$, the zero points of beamforming gain against $\left|\frac{1}{r_1}-\frac{1}{r_2}\right|$ remain unchanged and the height of side lobes will be reduced as $M$ increases. To show the correctness of our analysis, a simulation of beamforming with cylindrical arrays will be presented in the following section.
\section{Simulation Results}\label{sec: sim}
\par In this section, simulations are provided to verify our theoretical analysis of the beamforming with UCA and cylindrical arrays. We first assume that an $800$-element UCA is employed at BS with half-wavelength spacing and operating at $30$ GHz. Therefore the radius of UCA is about $0.64\,{\rm{m}}$, which is equivalent to the array aperture of a $256$-element ULA. The simulation parameters are summarized in Table~\ref{tab:1}.
\begin{table}[!t]
\renewcommand{\arraystretch}{1}
\caption{Simulation Parameters}
\label{tab:1}
\centering
\begin{tabular}{|*{2}{c|}}
\hline
The number of UCA antennas $N$ & 800\\
\hline
Carrier frequency $f$ & 30 GHz \\
\hline
Radius of UCA $R$ & 0.64 m \\
\hline
Spacing between antennas $d$ & 0.5 cm \\
\hline
Focal point ($r$, $\phi$) in Fig. 7 & (20 m, 0) \\
\hline
Distances of pair of focal points ($r_1$, $r_2$) in Fig. 8 & (20 m, 30 m) \\
\hline
Number of paths $L$ in Fig. 10 & 3 \\
\hline
Physical distances of users in Fig. 10 & $\mathcal{U} [4\,{\rm{m}},50\,{\rm{m}}]$ \\
\hline
\end{tabular}
\end{table}
\par First, we evaluate the effectiveness of estimations on the beamforming gain in the angular domain. The beamforming gain against spatial angles is plotted in Fig.~\ref{img: sim angular}, which indicates that estimations in {\bf{Lemma}~\ref{lemma1}} perfectly match the beamforming gain accurately calculated from the geometry relationships in~\eqref{eq: gain near}. In addition, it can also be seen that beamforming gain against directions is invariant with propagation distances, which verifies the accuracy of the assumption in {\bf{Lemma}~\ref{lemma1}}.

 \begin{figure}[!t]
    \centering
    \subfigure[Angular Domain]{
    	\includegraphics[width=3in]{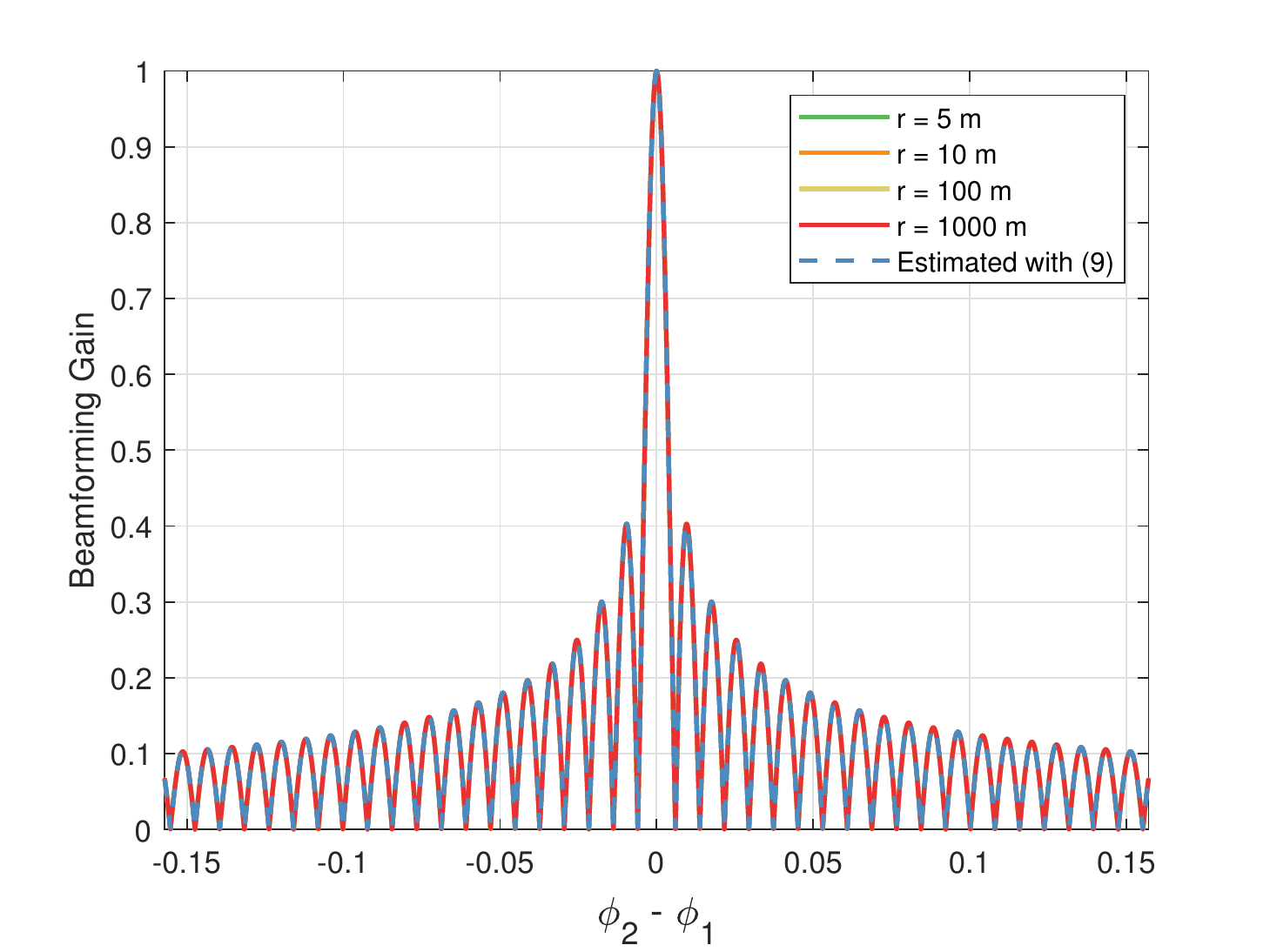}
    	\label{img: sim angular}}
    \subfigure[Distance Domain]{
    	\includegraphics[width=3in]{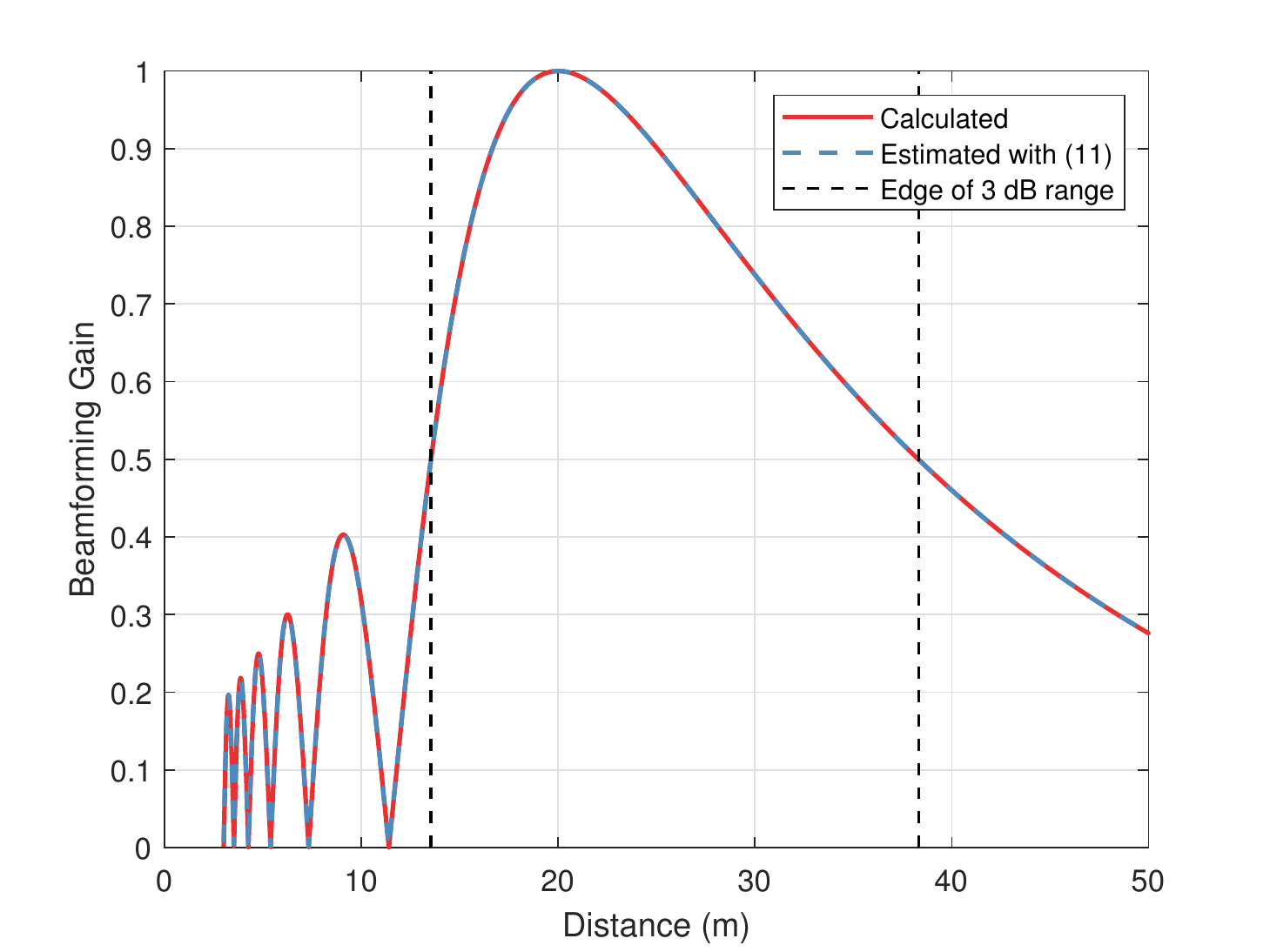}
    	\label{img: sim distance}}
    \caption{Verification of beamforming gain in angular and distance domains.}
    \label{img: beams ang dis focus}
    \vspace{-1em}
\end{figure}

\begin{figure}[!t]
    \centering
    \setlength{\abovecaptionskip}{0.cm}
    \includegraphics[width=3in]{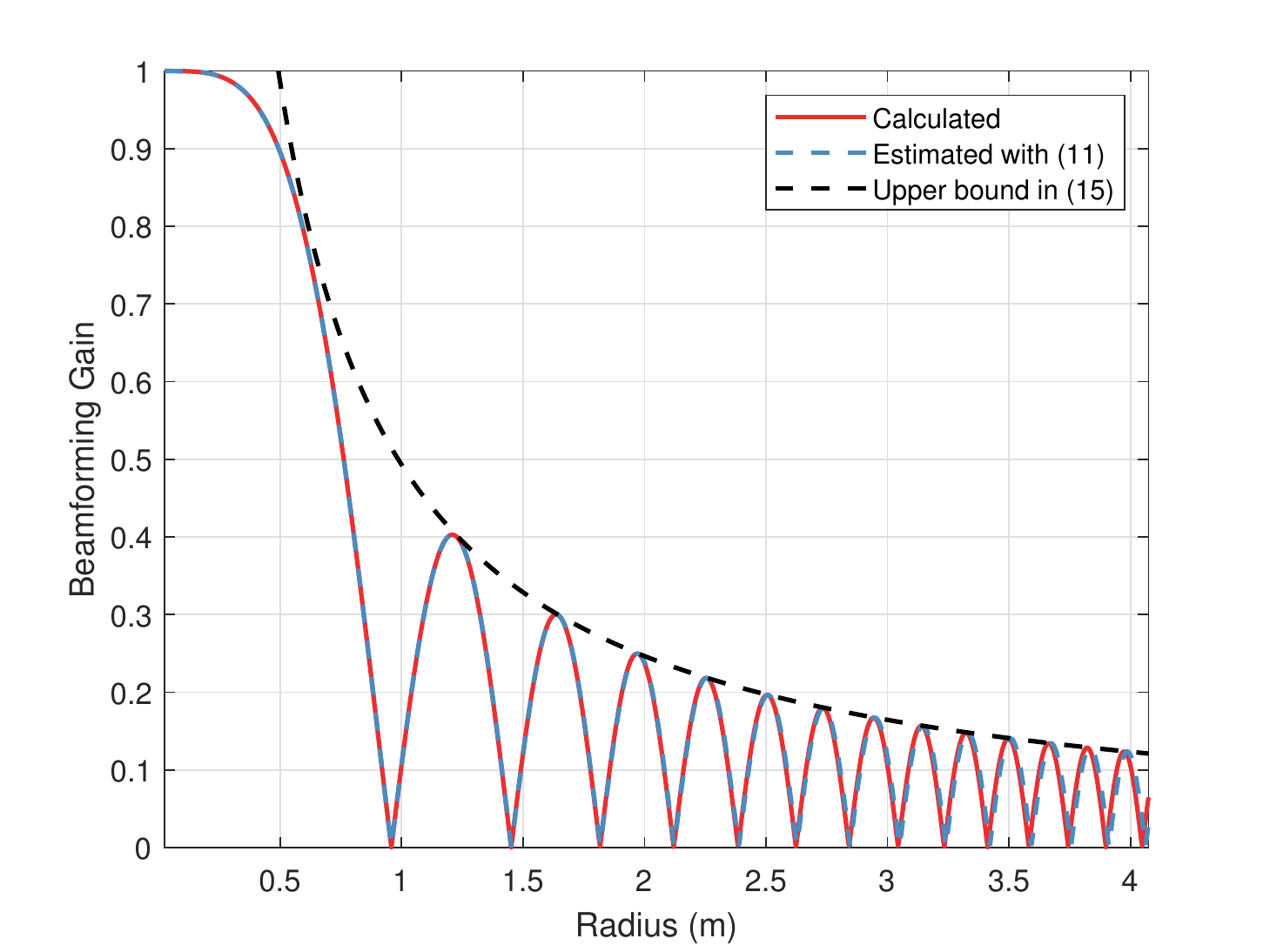}
    \caption{Verification of upper bound of beamforming gain.}
    \label{img: upper bound}
    \vspace{-1em}
\end{figure}

\par Then, the beamforming gain in the distance domain is verified in Fig.~\ref{img: sim distance}. The focal point is selected as $r = 20\,{\rm{m}}$ and $\phi = 0$. The estimated beamforming gain in {\bf{Lemma}~\ref{lemma2}} plotted in the blue line perfectly match the calculation results plotted in the red line. The depth-of-focus obtained in~\eqref{eq: beam depth} is also consistent with the simulation results.

\par To verify the upper bound derived in {\bf{Corollary}~\ref{coro2}}, we assume fixed $r_1 = 20\,{\rm{m}}$ and $r_2 = 30\,{\rm{m}}$, and plot the beamforming gain against the radius of UCA in Fig.~\ref{img: upper bound}. The upper bound accurately captures the descending trend of the envelope of  beamforming gain.

 \begin{figure}[!t]
    \centering
    \subfigure[Angular Samplings]{
    	\includegraphics[width=2.5in]{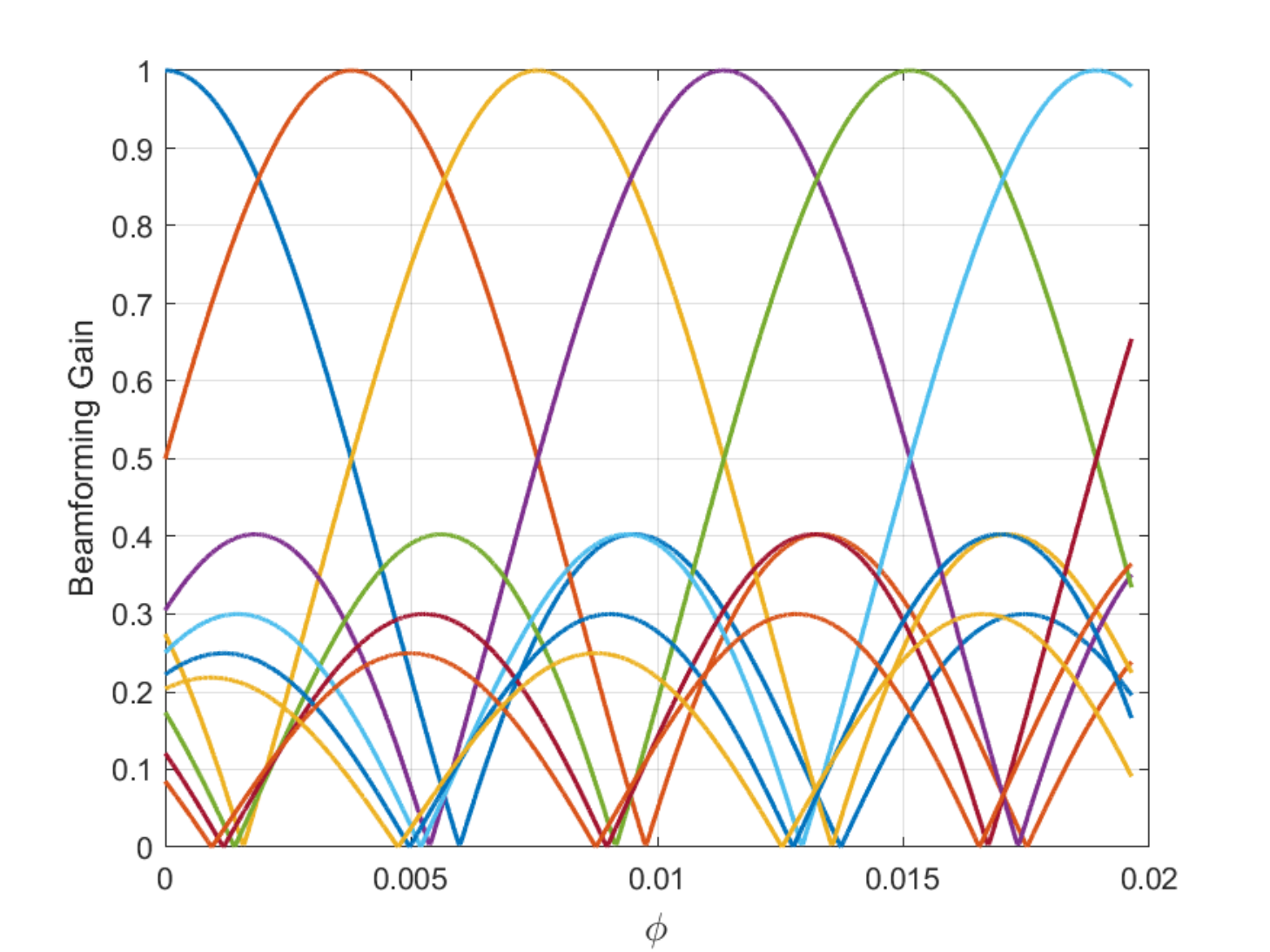}}
    \subfigure[Distance Samplings]{
    	\includegraphics[width=2.5in]{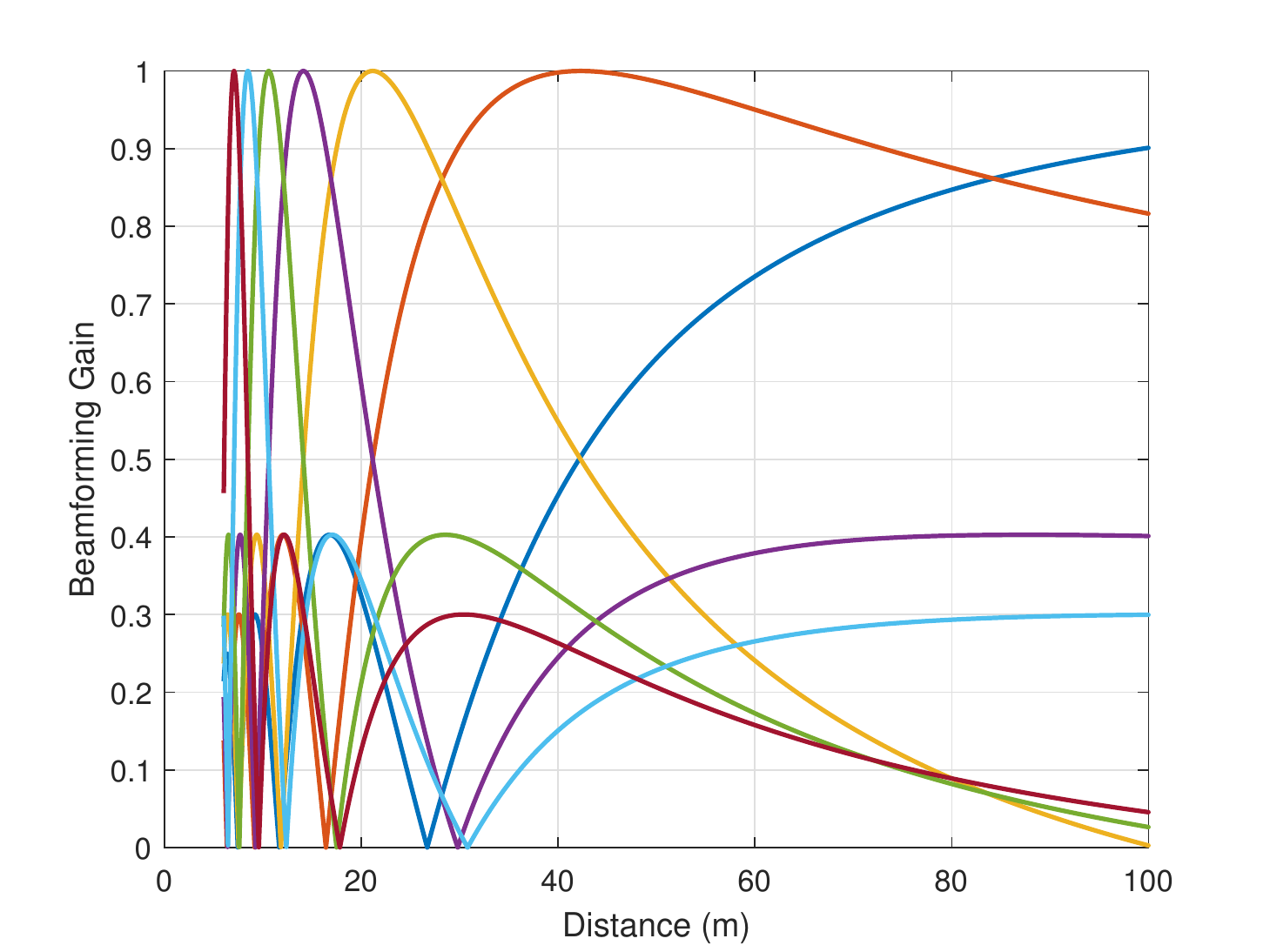}}
    \caption{Illustration of the sampling method in angular and distance domain.}
    \label{img: beams ang dis}
    \vspace{-1.5em}
\end{figure}

\par To verify the effectiveness of the proposed codebook design, the threshold is set to $\Delta = 0.5$ with minimum distance $r_{\rm{min}} = 4\,{\rm{m}}$. The beamforming gain for different codewords in both angular and distance domains is plotted in Fig.~\ref{img: beams ang dis}, where different colored lines denote the beamforming gain for different codewords. It can be seen that the peak of one codeword corresponds to the 3-{\rm{dB}} beamforming attenuation point of adjacent codewords, indicating that correlation of codewords does not exceed the predetermined threshold $\Delta = 0.5$.

\par To further verify the necessity of employing near-field beamforming in the near-field region, we compare the achievable rate $R = \log_2\left(1+\frac{P|{\bf{h}}_{\rm{near}}^H {\bf{w}}^*|^2}{\sigma_n^2}\right)$ with different beamforming codebooks, where $P$ and $\sigma_n$ denote the transmitted power and noise power, respectively. The wireless channel is modeled as in~\eqref{eq: near-field channel} with the number of paths $L=3$. The propagation distance in each path is assumed to be uniformly distributed within the range $[4\,{\rm{m}},50\,{\rm{m}}]$. The selected beamforming vector ${\bf{w}}^*$ is chosen from the corresponding codebook that maximizes the beamforming gain defined in~\eqref{eq: gain near}. The near-field concentric-ring codebook is employed in comparison to classical far-field codebooks as shown in Fig.~\ref{img: sum rate}. The proposed codebook could obtain $25\%$ performance gain compared with far-field codebooks and approach optimal beamforming in different SNRs. It is worth noting that, the performance gain originates from the increased beamforming gain in the near-field region. As the communication distances scale up, the performance of near-field codebook will degrade into the far-field counterpart.

\begin{figure}[!t]
    \centering
    \setlength{\abovecaptionskip}{0.cm}
    \includegraphics[width=3in]{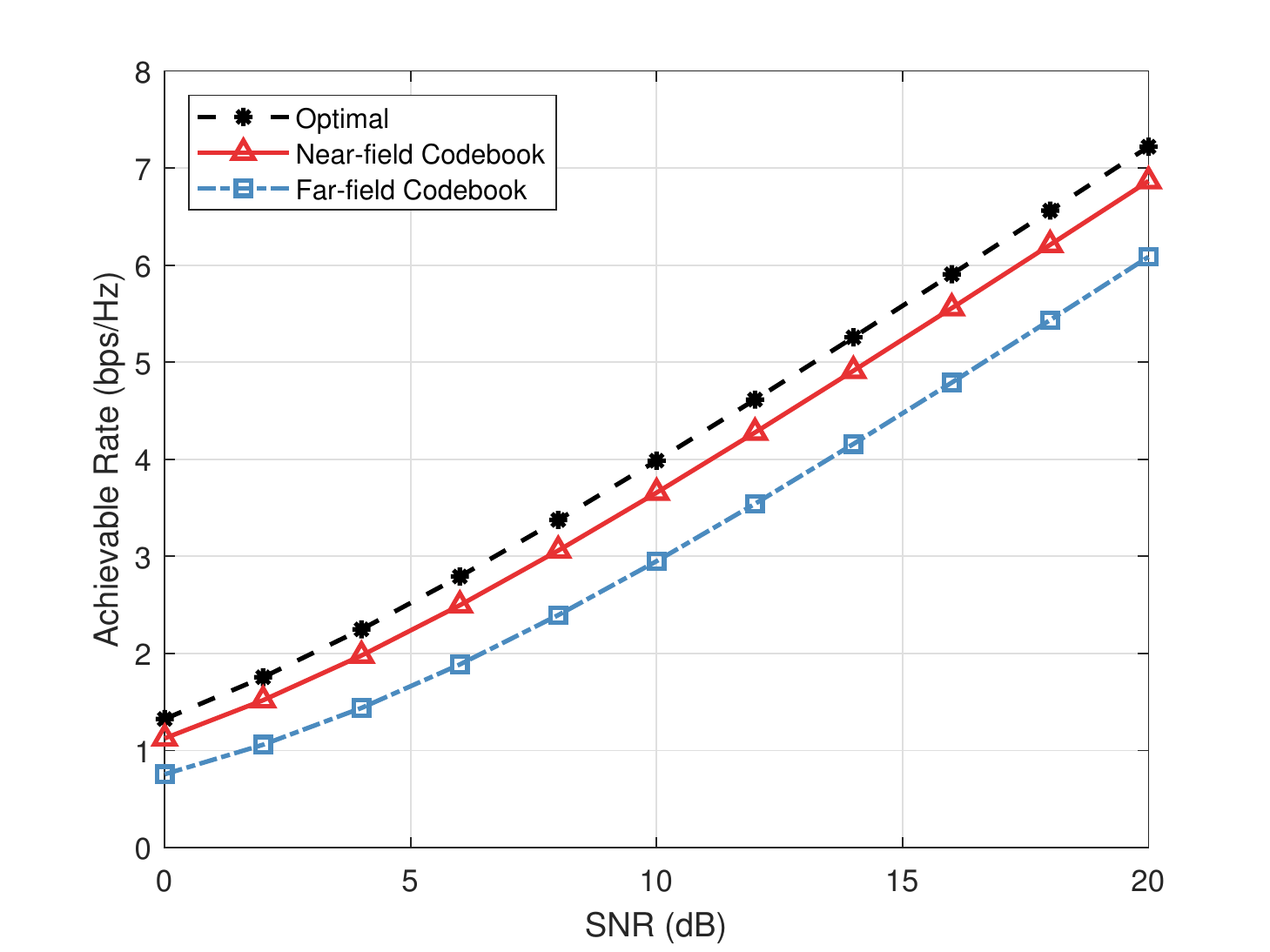}
    \caption{Comparison of the achievable sum rate.}
    \label{img: sum rate}
    \vspace{-1.5em}
\end{figure}

\par Another simulation is performed to verify the analysis for cylindrical arrays. The number of antennas on each UCA is set to $N=600$ with a varying number of concentric UCAs. The operating frequency is again set to $30$ GHz. Assuming the beam focusing vector ${\bf{b}}_{\rm{C}}(\infty, \pi/2, 0)$ is adopted, the beamforming gain against different distances is shown in Fig.~\ref{img: cylin gain}. It is shown that~\eqref{eq: gain cylin distance 2} well matches the simulations with different settings of $M$. In addition, the location of zeros remains unchanged and the height of side lobes could be suppressed by increasing the value of $M$, showing its additional potential to control the shape of side lobes compared with UCA.

\begin{figure}[!t]
    \centering
    \setlength{\abovecaptionskip}{0.cm}
    \includegraphics[width=3in]{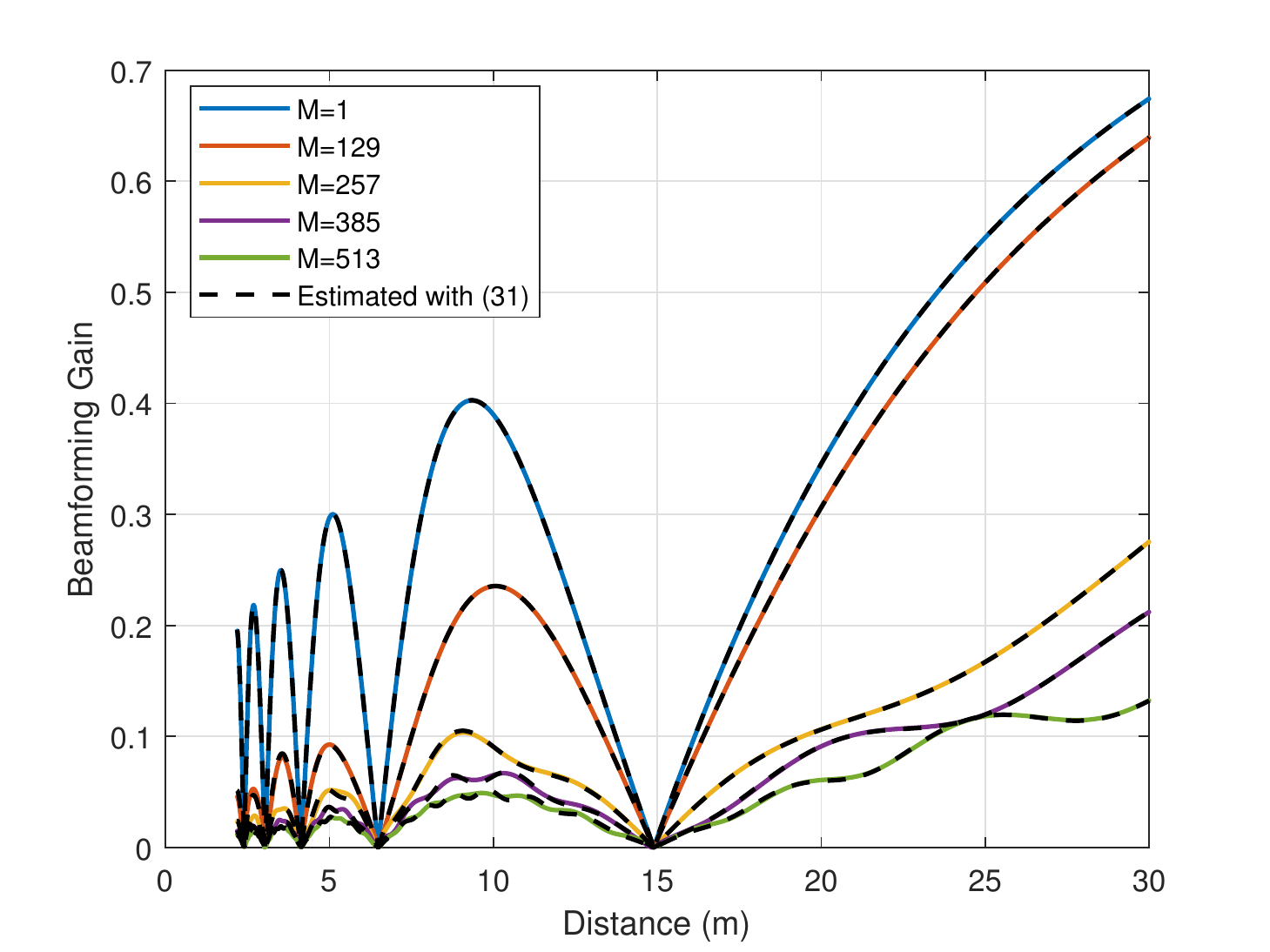}
    \caption{Beamforming gain against distances for different settings of $M$.}
    \label{img: cylin gain}
\end{figure}

\section{Conclusion}\label{sec: con}
In this paper, the near-field beamforming property of UCA is investigated by exploiting the array geometrical features. The focusing ability in the distance domain and its asymptotic property are theoretically characterized, indicating the potential of UCA to enlarge the near-field region compared with ULA. The concentric-ring codebook is proposed to enable codebook-based beamforming for near-field UCA communications. Then, the beamforming for a generalization form of UCA, cylindrical array, is briefly discussed. Simulation results verify the effectiveness of the theoretical analysis. The analysis of UCA beamforming presents a way to broaden the near-field region to enable communication technologies highly dependent on near-field propagation environments. In terms of future research, there is a need for further validation of the potential of UCA to effectively mitigate  interferences in multi-user communication scenarios. Apart from the cylindrical array as the 3D extension of circular arrays, spherical arrays are promising to avoid the decrease of effective aperture in the elevation angle to further enlarge the near-field region, which could be also investigated in the future. The blockage problem in UCAs and cylindrical arrays, where only a subset of antennas could contribute effectively to beamforming, is also a critical condition that requires further investigation.
\vspace{-1em}
\appendices
\section{Proof of Lemma~\ref{lemma1}}\label{app: lemma1}
With the assumption $r_1=r_2=r$, the beamforming gain can be expressed as
\begin{equation}
\label{eq: app1 1}
\begin{aligned}
g(r, \phi_1, r, \phi_2) & = \frac{1}{N} \left| \sum_{n=1}^{N} e^{j\frac{2\pi}{\lambda}\sqrt{r^2+R^2-2rR\cos(\phi_1 - \psi_n)}}\right.\\
&~~~~ \left. \times e^{-j\frac{2\pi}{\lambda}\sqrt{r^2+R^2-2rR\cos(\phi_2 - \psi_n)}} \right|\\
& \mathop{\approx}\limits^{(a)} \frac{1}{N} \left| \sum_{n=1}^{N} e^{-j\frac{2\pi}{\lambda}\left[R\cos(\phi_1 - \psi_n) - R\cos(\phi_2 - \psi_n)\right]} \right|\\
&\mathop{=}\limits^{(b)} \frac{1}{N} \left| \sum_{n=1}^{N} e^{j\beta \sin(\psi_n - \frac{\phi_1+\phi_2}{2})} \right|,
\end{aligned}
\end{equation}
where $\beta = \frac{4\pi R}{\lambda} \sin(\frac{\phi_2-\phi_1}{2})$. Approximation (a) is derived by only keeping the first-order of Taylor series expansion. Equation (b) is derived with the trigonometric identity. Then, we utilize Jacobi–Anger expansion of Bessel functions as~\cite{bowman2012introduction}
\begin{equation}
\label{eq: app1 2}
\begin{aligned}
e^{j\beta \cos\gamma} = \sum_{m=-\infty}^{\infty} j^m J_m(\beta) e^{jm\gamma},
\end{aligned}
\end{equation}
where $J_m(\cdot)$ denotes the $m$-order bessel function of the first kind. The beamforming gain can be rewritten as
\begin{equation}
\label{eq: app1 3}
\begin{aligned}
g(r, \phi_1, r, \phi_2) &\approx \frac{1}{N}\left|\sum_{n=1}^{N} \sum_{m=-\infty}^{\infty} j^m J_m(\beta) e^{jm(\frac{\pi}{2}-\psi_n+\frac{\phi_1+\phi_2}{2})} \right|\\
&= \frac{1}{N}\left|\sum_{m=-\infty}^{\infty} j^m J_m(\beta) e^{jm\frac{\phi_1+\phi_2+\pi}{2}} \sum_{n=1}^{N} e^{-jm\psi_n} \right|.
\end{aligned}
\end{equation}
Then, the summation over $n$ can be derived by the geometric progression as the piecewise function as
\begin{equation}
\label{eq: app1 4}
\sum_{n=1}^{N} e^{-jm\psi_n}=\left\{
\begin{aligned}
N, \quad &m = N \cdot t, t \in \mathbb{Z} \\
0, \quad &m \neq N \cdot t, t \in \mathbb{Z},\\
\end{aligned}
\right
.
\end{equation}
which means the summation equals to zero except $m$ is the integer multiple of the number of antennas $N$. According to the asymptotic property of Bessel function as
\begin{equation}
\label{eq: app1 5}
\begin{aligned}
\left| J_{|m|}(\beta) \right| \leq \left( \frac{\beta e}{2|m|} \right) ^{|m|},
\end{aligned}
\end{equation}
we can obtain the approximation $|J_{|m|}(\beta)| \approx 0$ if $m = N \cdot t$ with $t \neq 0$ for large $N$. Therefore, the summation could be well approximated with the term of $m=0$. Similar approximations were also adopted in~\cite{Zhang'17'tawp}. To sum up, the beamforming gain could be finally simplified as
\begin{equation}
\label{eq: app1 6}
\begin{aligned}
g(r, \phi_1, r, \phi_2) &\approx \left| J_0(\beta) \right|.
\end{aligned}
\end{equation}
In addition, the approximation error could be characterized by the first two terms as
\begin{equation}
\label{eq: app1 7}
\begin{aligned}
\Delta_g& = \left|g(r, \phi_1, r, \phi_2)-J_0(\beta)\right|\\
& \approx \left| j^{N} J_N(\beta) \cdot 2\cos\left(\frac{\phi_1+\phi_2+\pi}{2}N\right)  \right| \leq 2\left( \frac{\beta e}{2N} \right)^N,
\end{aligned}
\end{equation}
which is monotonically decreasing with $N$. Therefore, if we assume a large $N$, the approximation~\eqref{eq: app1 6} could achieve an accurate approximation, which completes the proof.

\section{Proof of Lemma~\ref{lemma2}}\label{app: lemma2}
Substituting the expression of $\xi_{r,\phi}^{(n)}$ into~\eqref{eq: gain near} with the condition $\phi_1 = \phi_2 = \phi$, we have 
\begin{equation}
\label{eq: app2 1}
\begin{aligned}
g(r_1, \phi, r_2, \phi) & \mathop{\approx}\limits^{(a)} \frac{1}{N} \left|\sum_{n=1}^{N} e^{j \frac{2\pi}{\lambda}\left\{\left(\frac{R^2}{4r_1}-\frac{R^2}{4r_2}\right)\left(1-\cos(2\phi-2\psi_n)\right)\right\} } \right|\\
& \mathop{=}\limits^{(b)} \frac{1}{N} \left| e^{j \frac{2\pi}{\lambda} \left(\frac{R^2}{4r_1}-\frac{R^2}{4r_2}\right) } \cdot \sum_{n=1}^{N} e^{j \bar{\zeta} \cos(2\phi-2\psi_n) } \right|\\
& = \frac{1}{N} \left|\sum_{n=1}^{N} e^{j \bar{\zeta} \cos(2\phi-2\psi_n) } \right|,
\end{aligned}
\end{equation}
where approximation (a) is obtained by trigonometric functions $\cos^2(x) = \frac{1+\cos(2x)}{2}$. Equation (b) is derived by assuming $\bar{\zeta} = \frac{2\pi}{\lambda}\left(-\frac{R^2}{4r_1}+\frac{R^2}{4r_2}\right)$. Again, our analysis hinges on the Jacobi–Anger expansion in~\eqref{eq: app1 2}. Thus, the beamforming gain can be rewritten as
\begin{equation}
\label{eq: app2 3}
\begin{aligned}
g(r_1, \phi, r_2, \phi) &\approx \frac{1}{N}\left|\sum_{n=1}^{N} \sum_{m=-\infty}^{\infty} j^m J_m(\bar{\zeta}) e^{jm(2\phi-2\psi_n)} \right|\\
& = \frac{1}{N}\left|\sum_{m=-\infty}^{\infty} j^m J_m(\bar{\zeta}) e^{2jm\phi} \sum_{n=1}^{N} e^{-2jm\psi_n} \right|\\
& \mathop{\approx}\limits^{(c)} \left| J_0(\bar{\zeta}) \right|,
\end{aligned}
\end{equation}
where approximation (c) is obtained similar to Appendix~\ref{app: lemma1}.
\par In addition, with the property of the even function $J_0(\cdot)$, $J_0(\bar{\zeta}) = J_0(\zeta)$ can be ensured with $\zeta = |\bar{\zeta}| = \frac{2\pi}{\lambda}\left|\frac{R^2}{4r_1}-\frac{R^2}{4r_2}\right|$, which completes the proof.



\section{Proof of Lemma~\ref{lemma4}}\label{app: lemma4}
Following the definition of beamforming gain in~\eqref{eq: gain cylin distance} with the assumption $\theta = \pi/2$, we can obtain $\chi_3 = 0$. Then, the beamforming gain could be simplified by switching the order of summation as
\begin{equation}
\label{eq: app4 1}
\begin{aligned}
g_{\rm{C}} &= \left| \frac{1}{N} \sum_{n=1}^{N} e^{j\frac{\pi R^2}{\lambda}\left(\frac{1}{r_1}-\frac{1}{r_2}\right)  \left(1-\cos^2(\phi-\psi_n) \right)}\right.\\
& ~~~~ \left. \times \frac{1}{{\bar{M}}} \sum_{m=-M}^{M} e^{j\frac{\pi m^2d^2}{\lambda}\left(\frac{1}{r_1}-\frac{1}{r_2}\right)} \right|,
\end{aligned}
\end{equation}
where the variables ${\bar{\Theta}}_1$ and ${\bar{\Theta}}_2$ are neglected for expression simplicity. Note that the first summation over $n$ could be addressed similarly as in~{\bf Appendix~\ref{app: lemma2}}, we turn to analysis of the second summation over $m$.
\par The second summation is much like the beamforming analysis for ULA systems in~\cite{Cui'22'tcom}. By denoting $a = \sqrt{\frac{d^2}{\lambda}(\frac{1}{r_1}-\frac{1}{r_2})}$ and assuming $r_1 \leq r_2$, the summation could be rewritten as
\begin{equation}
\label{eq: app4 2}
\begin{aligned}
&~~~~\frac{1}{{\bar{M}}} \sum_{m=-M}^{M} e^{j \pi (am)^2} \\
& \mathop{\approx}\limits^{(a)} \frac{1}{M} \int_{0}^{M} e^{j \pi (am)^2} {\rm{d}}m\\
& = \frac{1}{\sqrt{2}aM} \int_{0}^{\sqrt{2}aM} e^{j \frac{\pi}{2}t^2} {\rm{d}}t\\
& = \frac{\int_{0}^{\sqrt{2}aM} \cos(\frac{\pi}{2}t^2) {\rm{d}}t +j\int_{0}^{\sqrt{2}aM} \sin(\frac{\pi}{2}t^2) {\rm{d}}t}{\sqrt{2}aM}\\
& \mathop{=}\limits^{(b)} \frac{C(\mu)+jS(\mu)}{\mu} = G(\mu),
\end{aligned}
\end{equation}
where the approximation (a) is derived by replacing summation with integration and $\mu = \sqrt{2}aM = \sqrt{\frac{2M^2d^2}{\lambda}(\frac{1}{r_1}-\frac{1}{r_2})}$. Equation (b) is derived by denoting the Fresnel functions as $C(\mu) = \int_{0}^{\mu} \cos(\frac{\pi}{2}t^2) {\rm{d}}t$ and $S(\mu) = \int_{0}^{\mu} \sin(\frac{\pi}{2}t^2) {\rm{d}}t$~\cite{Sherman'62'j}.
\par Moreover, if $r_1 > r_2$ we can define $a = \sqrt{\frac{d^2}{\lambda}(\frac{1}{r_2}-\frac{1}{r_1})}$ and the conclusion above still holds. Therefore, we can define $\mu = \sqrt{\frac{2M^2d^2}{\lambda}\left| \frac{1}{r_1}-\frac{1}{r_2} \right|}$, leading to the beamforming gain as
\begin{equation}
\label{eq: app4 3}
\begin{aligned}
g_{\rm{C}} = \left| G(\mu) J_0(\zeta)\right|,
\end{aligned}
\end{equation}
where $\zeta = \frac{2\pi R^2}{\lambda}\left|\frac{1}{4r_1}-\frac{1}{4r_2}\right|$. This completes the proof.
\footnotesize
\balance 
\bibliographystyle{IEEEtran}
\bibliography{IEEEabrv,reference}

\end{document}